\documentclass[11pt]{article}

\usepackage{geometry}
\usepackage[T1]{fontenc}
\usepackage{amsmath}
\usepackage{amssymb}
\usepackage{amsthm}
\usepackage{enumitem}
\usepackage{hyperref}
\usepackage{cleveref}
\usepackage{subcaption}
\usepackage{float}
\usepackage{mathtools}
\usepackage[ruled]{algorithm2e} % For algorithms

\usepackage{graphicx}
\usepackage{xcolor}
\usepackage{url}
\usepackage{diagbox}
\usepackage[comma, numbers]{natbib}
\usepackage{bm}

\usepackage{amsthm}

\newtheorem{definition}{Definition}[section]
\newtheorem{theorem}{Theorem}[section]
\newtheorem{lemma}{Lemma}[section]
\newtheorem{example}{Example}[section]

\newtheorem{assumption}{Assumption}[section]
\newtheorem{observation}{Observation}[section]

\newtheorem{claim}{Claim}[section]
\newtheorem{proposition}{Proposition}[section]
\newtheorem{remark}{Remark}[section]
\usepackage{algorithmic}

\def\+#1{\mathcal{#1}}
\def\-#1{\mathbb{#1}}

\newcommand{\notshow}[1]{{}}
\newcommand{\AutoAdjust}[3]{{ \mathchoice{ \left #1 #2  \right #3}{#1 #2 #3}{#1 #2 #3}{#1 #2 #3} }}
\newcommand{\Xcomment}[1]{{}}

\newcommand{\InParentheses}[1]{\AutoAdjust{(}{#1}{)}}
\newcommand{\InBrackets}[1]{\AutoAdjust{[}{#1}{]}}

\usepackage{amsmath,amsfonts,bm}

\def\eqref#1{equation~\ref{#1}}
\def\1{\mathbb{I}}

\def\eps{{\varepsilon}}

\DeclareMathAlphabet{\mathsfit}{\encodingdefault}{\sfdefault}{m}{sl}
\SetMathAlphabet{\mathsfit}{bold}{\encodingdefault}{\sfdefault}{bx}{n}

\newcommand{\E}{\mathbb{E}}

\newcommand{\R}{\mathbb{R}}

\DeclareMathOperator*{\argmax}{arg\,max}

\newcommand{\truth}{\mathrm{Truth}}

\newcommand{\gyr}[1]{{\color{purple} [Yanru: #1]}}

\title{On the Coordination of Value-Maximizing Bidders\thanks{Accepted at ICML 2026.}}
\author{
Yanru Guan\\
Peking University\\
\texttt{piscesguan@stu.pku.edu.cn}
\and
Jiahao Zhang\\
Carnegie Mellon University\\
\texttt{jiahaozhang@cmu.edu}
\and
Zhe Feng\\
Google DeepMind\\
\texttt{zhef@google.com}
\and
Tao Lin\\
Microsoft Research \& The Chinese University of Hong Kong, Shenzhen\\
\texttt{lintao@cuhk.edu.cn}
}
\date{}

\begin{document}

\maketitle


\begin{abstract}
While the auto-bidding literature predominantly considers independent bidding, we investigate the coordination problem among multiple auto-bidders in online advertising platforms.
Two motivating scenarios are: collaborative bidding among multiple bidders managed by a third-party bidding agent, and strategic bid selection for multiple ad campaigns managed by a single advertiser.
We formalize this coordination problem as a theoretical model and investigate the coordination mechanism where only the highest-value bidder competes with outside bidders, while other coordinated bidders refrain from competing. 
We demonstrate that such a coordination mechanism dominates independent bidding, improving both Return-on-Spend (RoS) compliance and the total value accrued for the participating auto-bidders or ad campaigns, for a broad class of auto-bidding algorithms. 
Additionally, our simulations on synthetic and real-world datasets support the theoretical result that coordination outperforms independent bidding. 
These findings highlight both the theoretical potential and the practical robustness of coordinated auto-bidding in online auctions.
\end{abstract}

% \keywords{auto-bidding, coordination, return on spend, value maximization}


\section{Introduction}
Online advertising is a prominent approach to monetize search engines, news feeds, social media, and e-commerce. Among them, real-time auctions have been widely used to connect advertisers and users efficiently. In modern online advertising platforms, advertisers endeavor to maximize their campaign effectiveness, quantified by conversions or other pertinent metrics, through precise bid management that considers targeted Return-on-Spend (RoS)~\citep{aggarwal2019autobidding}, which is also well-known as auto-bidding problem. To this end, a wide spectrum of bidding methodologies has been formulated, leveraging concepts from optimization theory, online learning paradigms, and game-theoretic principles~\citep{zhao2018deep, aggarwal2019autobidding, lee2013real,babaioff2021non, golrezaei2021bidding, deng2021towards,   balseiro2021landscape, gao2022bidding,stram2024mystique}.

Despite the wealth of sophisticated research in online advertising, existing literature predominantly concentrates on optimizing the value or utility from the perspective of \emph{individual} bidders. % operating under RoS constraints.
This focus, while valuable, often overlooks a crucial aspect of real-world advertising dynamics: the possibility of \emph{coordinated} bidding strategies. In practice, it is common for multiple bidders to form coalitions, perhaps by authorizing a third-party advertising platform to manage their bids collectively. Alternatively, a single large advertiser, such as major e-commerce players like Amazon, Temu, or Shein, might control numerous distinct advertising campaigns. These entities often have a portfolio of advertisements they are willing to display to users, necessitating a coordinated approach across their various campaigns to maximize overall effectiveness, rather than optimizing each campaign in isolation.
Consequently, a deeper understanding of how auto-bidding algorithms can be coordinated across multiple allied bidders or campaigns, and what are the benefits of coordination over independent bidding, is becoming increasingly critical.
% to reflect these prevalent industry practices and unlock further efficiencies.
Such insights might be useful for industry practices and unlock further efficiencies in online advertising markets.

In this work, we investigate the coordination problem among multiple auto-bidding algorithms in repeated auctions.  
% To address this practical problem,
We formulate a theoretical model in which $N$ auto-bidders participate in $T$ rounds of second-price auctions.
The values of these bidders are independent and identically distributed (i.i.d.), motivated by the practice that only similar bidders are chosen to compete in ad auctions. 
These bidders may form a coalition to determine their bids jointly, competing against bidders outside the coalition.
Each auto-bidder is modeled as a value-maximizing agent subject to an RoS constraint (i.e., non-negative overall utility).
We consider a simple form of coordination: only the bidder with the highest value inside the coalition submits a positive bid to compete against the outside bidders, while the other $N-1$ coalition members refrain from competing. 
In contrast, in the independent bidding scenario, all the $N$ bidders submit positive bids to compete against each other and the outside bidders. We aim to compare the coalition bidders' welfare in the coordinated and independent bidding scenarios.

The main contributions of this paper are as follows. 
Theoretically, we prove that the highest-value-bidder-compete coordination mechanism mentioned above, although straightforward, enjoys a significant advantage over independent bidding: \emph{By coordination, bidders in the coalition acquire higher values, as well as achieve lower RoS constraint violations, compared to independent bidding, for a broad class of auto-bidding algorithms.}

In particular, we identify a condition on the distributions of bidders' values and outside bids -- Assumption~\ref{ass:highest-value-assumption} -- under which each coalition bidder's total utility in $T$ periods is higher, hence the bidder's RoS constraint is less violated, under coordination than under independent bidding (Theorem \ref{thm:utility/RoS-constraint}). This conclusion holds for any auto-bidding algorithm that overbids (a common feature for value-maximizing algorithms).  Assumption~\ref{ass:highest-value-assumption} is necessary and sufficient: when it does not hold, there exists an overbidding auto-bidding algorithm for which coordination hurts the coalition bidders, reducing their utility (RoS constraint compliance). 

% Concretely, when Assumption~\ref{ass:highest-value-assumption} holds, every bidder's expected cumulative utility increases by at least $T\Delta/N$ relative to independent bidding; conversely, when Assumption~\ref{ass:highest-value-assumption} is strictly violated, there exists an overbidding algorithm for which coordination becomes worse for all bidders for sufficiently large horizons (Theorem~\ref{thm:utility/RoS-constraint}).

% \tao{Edited until here.}

% For the state-of-the-art RoS auto-bidder of \citet{feng2023online}, we prove a stronger statement than “beats independent bidding”:
% under Assumption~\ref{ass:highest-value-assumption}, coordinated bidding asymptotically matches the best possible coalition value among \emph{all} coordination algorithms (Theorem~\ref{thm:value}).
% In particular, the coordinated coalition attains the information-theoretic upper bound $\mathbb{E}[v^{(N)}]$ in the long run.

% We further interpret \citet{feng2023online}'s update as mirror descent with an entropy mirror map and extend the analysis to a general MD-$h$ family.
% For this broader class, we show that coordination weakly dominates \emph{independent} bidding in asymptotic average value (Theorem~\ref{thm:value-wo-ass}); all technical proofs are deferred to the appendix.

We then consider the total value of coalition bidders. We prove that, as long as bidders use \emph{mirror-descent} algorithms -- a broad class of state-of-the-art auto-bidding algorithms under RoS constraints -- then the total value of the coalition bidders will be improved by coordination, regardless of whether Assumption \ref{ass:highest-value-assumption} holds (Theorem \ref{thm:value-wo-ass}).
If Assumption \ref{ass:highest-value-assumption} holds, then such coordinated mirror-descent algorithms not only dominate independent bidding but also achieve the best possible coalition value among \emph{all} possible coordination mechanisms, as $T\to\infty$ (Theorem \ref{thm:value}).

Beyond the i.i.d.~setting, we extend our results to cases where coalition bidders' values are independent but not identically distributed. In this non-i.i.d.~setting, we demonstrate that -- under slightly stronger assumptions -- the same simple coordination algorithm still yields better total value and total RoS violation compared to independent bidding. But in contrast to the i.i.d.~case, we cannot guarantee that this algorithm improves the performance for every individual bidder. Designing a coordination mechanism that achieves superior individual results in the non-i.i.d.~setting remains an interesting open problem.

Finally, Section \ref{sec:experiments} demonstrates, by experiments on synthetic and real-world datasets, that our coordinated mechanism consistently outperforms independent bidding, validating our theoretical results under realistic market conditions.

% First, in the setting where the valuations of all auto-bidders within the coalition are independent and identically distributed (i.i.d.), we prove that a straightforward coordination mechanism strictly improves upon independent bidding strategies, both in terms of value acquisition and adherence to RoS constraints.

% Second, for the i.i.d and non-i.i.d. setting, we show that experiments on both synthetic and real-world datasets demonstrate that the coordinated mechanism consistently outperforms independent bidding in realistic market conditions.

% The main contributions of this paper are twofold. Firstly, in the setting where the valuations of all auto-bidders within the coalition are independent and identically distributed (i.i.d.), we demonstrate the existence of a straightforward coordination mechanism. This mechanism is proven to yield superior performance for every auto-bidder in the coalition, with respect to both value acquisition and RoS constraint adherence, when contrasted with independent bidding strategies (where each auto-bidder bids individually to optimize their respective objectives). Conversely, for the non-i.i.d. setting, we present a counter-example illustrating that the proposed coordination mechanism can, in fact, result in outcomes inferior to those achieved through independent bidding.

\subsection{Related Work}
Bidder coordination, often referred to as cartels or collusion in economics, has been extensively studied in traditional auction theory \citep{robinson1985collusion,hendricks1989collusion,marshall2007bidder}. Prior research has studied collusive behavior in first-price auctions \citep{lopomo2011bidder,pesendorfer2000study}, second-price auctions \citep{mailath1991collusion,graham1987collusive}, collusion-proof mechanism design \citep{che2009optimal}, auction design under collusion \citep{pavlov2008auction}, and collusion detection \citep{chotibhongs2012analysis}. While these studies typically consider bidders forming collusions \emph{autonomously} and reaching equilibrium outcomes in static auction settings \citep{bergemann2016bayes,fu2025learning}, we instead introduce a \emph{central planner} that coordinates bidders to improve their total welfare in dynamic auction environments.

Recently, coordination behavior has also been observed in online advertising auctions \citep{decarolis_bid_2023}, sparking growing interest in understanding and analyzing such phenomena. \citet{decarolis2020marketing} studied the impact of coordination on the revenue and efficiency of GSP and VCG auctions. \citet{romano2022power} investigated the computational challenges faced by a media agency coordinating bidders under those mechanisms. \citet{chen_coordinated_2023} analyzed coordinated online bidding in repeated second-price auctions with budget constraints.
% for \emph{utility maximizers}.
While those works focus on coordinating \emph{utility maximizers}, we study \emph{value maximizers}. To the best of our knowledge, this is the first work to investigate coordination among value maximizers, which exhibit different properties from utility maximizers.

Our work is also related to the large literature on auto-bidding. For example, \citet{paes2024complex} studied the dynamics of systems with multiple auto-bidders and revealed complex behaviors such as bi-stability, periodic orbits, and quasi-periodicity, while \citet{aggarwal2025multi} analyzed optimal bidding strategies for advertisers across multiple platforms. For a comprehensive overview of the literature, we refer readers to the survey by \citet{aggarwal_auto-bidding_2024}.
Our study complements the predominant individual-optimization perspective of the auto-bidding literature. 

% \tao{Maybe talk about the traditional econ literature on bidder collusion first. Then transition to coordination in ad auctions? }

\section{Model: Coordinated Auto-Bidding}
\label{sec:model}

\paragraph{Repeated Second-Price Auctions with Coalition.}
We consider the scenario where $N$ \emph{auto-bidders} (or \emph{bidders}) participate in $T$ rounds of repeated second-price auctions.
In each round $t\in [T] = \{1, \ldots, T\}$, one item (e.g., ad slot) is for sale, and each bidder $i \in [N] = \{1, \ldots, N\}$ draws a value $v_{i,t} \in [0, B]$ for the item i.i.d.~from a continuous distribution $F$ with full support on $[0, B]$.
Let $\bm v_t = (v_{1,t}, \ldots, v_{N,t})$ be the vector of values of all bidders at round $t$. 
Let $b_{i,t}$ denote bidder $i$'s bid at round $t$.
The $N$ bidders may form a coalition.
Our model captures two possible types of coalitions in ad auctions:
\begin{enumerate}
    \item $N$ advertisers, each with an ad, instruct their auto-bidding algorithms to coordinate with each other. % before bidding.
    \item One advertiser owns $N$ ads. Instead of using $N$ auto-bidding algorithms to bid for the $N$ ads independently, the advertiser instructs the algorithms to coordinate.
\end{enumerate}
In either case, usually in practice, only similar advertisers/ads are selected to compete for an ad slot, which justifies the i.i.d.~value assumption.
In addition, we will show in Section \ref{sec:asymmetric} that our results extend to non-i.i.d.~value settings.

There is a competing bid $d_t^O \in [0, B]$ from outside of the coalition, assumed to be i.i.d.~sampled from another distribution $D$ every round. Our model also captures second-price auctions with reservation price, because $d_t^O$ can be the maximum of the reservation price and the outside competing bid. 
% \paragraph{Second-price auction}
We denote the competing bid faced by bidder $i$ at round $t$ as 
\begin{equation}
    d_{i,t}=\max\big\{d^O_t, \max_{j\in [N]\setminus \{i\} }b_{j,t}  \,\big\}.
\end{equation}
Let $x_{i,t}:=\1\{b_{i,t}\ge d_{i,t}\}\in\{0,1\}$ indicate whether bidder $i$ wins the item (ignoring tie-breaking for now), so bidder $i$'s payment is $p_{i,t}:= x_{i,t}d_{i,t}$ and utility is $u_{i,t}:= x_{i, t} v_{i, t} - p_{i, t} = x_{i,t}(v_{i,t} - d_{i,t})$ at round $t$.

% Let $d_{t, i} = \max_{j\ne i} b_{t, j}$ be the competing bid for bidder $i$.  Let $x_{t, i} = \1\{b_{t, i}\ge d_{t, i}\}\in\{0,1\}$ denote whether bidder $i$ wins the item at round $t$. (We ignore tie-breaking for now.)  Let $p_{t, i} = x_{t, i}d_{t, i}$ denote bidder $i$'s payment at round $t$. 

%====== 
% We consider the scenario in which $N$ autobidders participate in $T$ rounds of repeated second price auctions. In each round $t=1, 2, \cdots, T$, there is an available ad slot.
% There is a bidder with $K$ ads $\gK=\{1,2,\cdots,K\}$ to display. The value of ad $k\in\gK$ at round $t$ is $v_{kt}$, which is i.i.d.~sampled from a distribution $F_k$. 
% 
% For ad $k\in\gK$, denote the bidder's bid at round $t$ by $b_{kt}$. For other bidders, we assume that the highest bid in round $t$, denoted by $d^O_t$, is sampled i.i.d.~from a distribution $H$. Thus the highest competing bid that ad $k$ faces is
% \begin{equation}
% d_{kt}=\max\{\max_{i\in\gK:i\neq k}b_{it},d^O_t\}.
% \end{equation}
% Let $x_{kt}:=\1\{b_{kt}\ge d_{kt}\}\in\{0,1\}$ indicate whether the bidder wins the ad slot and display ad $k$ at round $t$.
% We denote by $u_{kt}:=x_{kt}(v_{kt}-d^O_t)$ the utility that the bidder gains from ads $k$ at round $t$ and by $z_{kt}=x_{kt}d^O_t$ the corresponding cost for a second price auction.

\paragraph{Value Maximization under RoS Constraint.}
Each auto-bidder aims to maximize its total value while adhering to the Return on Spend (RoS) constraint. The formal definition of the optimization problem for bidder $i$ is as follows:
\begin{align*}
    \max_{ \{b_{i, t}\}_{t=1}^T } & \sum \nolimits_{t=1}^T v_{i,t}\cdot x_{i,t} & \text{(bidder $i$'s value)}\\
    \text{s.t.} \hspace{0.8em} & \sum \nolimits_{t=1}^T \Big( \underbrace{ v_{i,t}\cdot x_{i,t} - p_{i,t} }_{u_{i, t}}\Big)
    \ge 0 & \text{(RoS constraint)}. 
\end{align*}
The RoS constraint ensures that for every dollar spent, at least one dollar of value is generated; it is equivalent to requiring the bidder's total utility to be non-negative.\footnote{In fact, the general RoS constraint requires $\sum_{t=1}^T v_{i, t} \cdot x_{i, t} - \tau \cdot \sum_{t=1}^T p_{i, t} \geq 0$, i.e., at least $\tau$ dollar of value will be generated for every dollar spent.  However, it is without loss of generality to assume $\tau = 1$ in the above optimization problem % since it won't change the solution of this problem.
because we can rescale the bidders' values by dividing by $\tau$. 
}
% \begin{align}
% \sum_{i=1}^T \big( \underbrace{v_{i,t}\cdot x_{i,t} - p_{i, t}}_{u_{i, t}}\big) \ge 0 \quad \text{(equivalent RoS constraint)}
% \end{align}

\paragraph{Independent and Coordinated Auto-Bidding.}
Each auto-bidder runs an algorithm $A$ that learns to bid from history, in order to maximize the bidder's total value subject to RoS constraint. Formally, at each round \( t \), bidder $i$'s algorithm $A(\cdot)$ takes the
% bidding algorithm \( A \) takes the entire history \( \{\bm b_{t'}, \bm p_{t'}, \bm x_{t'}\}_{t'=1}^{t-1} \) as input and outputs the bid \( b_t \).
historical information $H_{i, t}$ available to bidder $i$ -- for example, the past values, bids, and allocations of the bidder: $H_{i, t} = \{v_{i, t'}, b_{i, t'}, x_{i, t'}\}_{t'=1}^{t-1} \cup\{v_{i, t}\}$ -- as input
and outputs the bid $b_{i, t}$ for the current round.
We refer to the scenario where the auto-bidding algorithms are run independently as \emph{independent bidding}, formalized below:  

\begin{algorithm}[H]
\caption{Independent Bidding}
\label{alg:independent-bidding}
\begin{algorithmic}
\FOR{$t=1, 2, \dots, T$}
\STATE \hspace{-0.7em} Each bidder $i \in [N]$ observes their value \( v_{i,t} \sim F \) and bids \( b_{i,t} = A(H_{i, t})\).
\STATE \hspace{-0.7em} Each bidder $i \in [N]$ obtains allocation \( x_{i,t} \) and pays \( p_{i,t} \).
\ENDFOR
\end{algorithmic}

\end{algorithm}

Next, we define a \emph{coordinated bidding} scenario where, at each round, only the highest-value bidder in the coalition places a positive bid while all other bidders bid zero.
We assume that ties are resolved naturally.\footnote{Ties happen with probability $0$ given a continuous $F$.} % value distribution $F$.} 
This is formalized in Algorithm \ref{alg:coordinated-bidding}. 

\begin{algorithm}[H]
\caption{Coordinated Bidding}
\label{alg:coordinated-bidding}
\begin{algorithmic}
\FOR{$t=1,2,\dots,T$}
\STATE Each bidder \( i \in [N] \) observes their value \( v_{i,t} \sim F\).
\STATE Let \( i^* = \arg\max_{i\in[N]} v_{i,t} \).
\STATE Bidder \(i^* \) bids \( b_{i^*,t} = A( H_{i^*, t}) \). 
\STATE Bidder \( i \neq i^* \) bids \( b_{i,t} = 0 \).
\STATE Each bidder $i \in [N]$ obtains allocation \( x_{i,t} \) and pays \( p_{i,t} \).
\ENDFOR
\end{algorithmic}
\end{algorithm}

% \textcolor{red}{TODO: Write a formal definition of general autobidding algorithms. Collusion chooses the bidder with the highest value.}

\section{Coordination Reduces RoS Violation}
\label{sec:RoS}

In this section, we analyze the extent to which the RoS constraint of every auto-bidder is maintained/violated.
% We will show that, under mild assumptions, coordinated bidding results in a smaller RoS violation compared to independent bidding.
We will identify the necessary and sufficient condition under which coordinated bidding results in a smaller RoS violation compared to independent bidding (equivalently, every bidder in the coalition obtains a higher utility through coordination), for any auto-bidding algorithm. 

% The technical details involve analyzing each bidder's total utility (the negation of RoS violation) and comparing the coordination case with the truthful independent bidding case.

% The conclusion of this section is also crucial for the analysis of the total value of the coordinated bidders in \cref{sec:value}.

Formally, let $U_{i, T}^{C, A}, U_{i, T}^{I, A}$ be the total utility of bidder $i$ during the $T$ rounds using auto-bidding algorithm $A$ under coordinated and independent bidding scenarios, respectively: 
\begin{equation*}
    U_{i, T}^{C, A} = \sum \nolimits_{t=1}^T u_{i, t}^{C, A} \qquad  U_{i, T}^{I, A} = \sum \nolimits_{t=1}^T u_{i, t}^{I, A}. 
\end{equation*}
RoS constraint requires $U_{i, T}^{C \text{ or } I, A} \ge 0$. If $U_{i, T}^{C \text{ or } I, A} \ge - c$, we say the RoS constraint for bidder $i$ is violated by $c$.

We allow the bidders to use \emph{any} auto-bidding algorithm that \emph{(weakly) overbids}: namely, $A(\cdot)$ can be any function that outputs bid $b_{i, t} = A(H_{i, t}) \ge v_{i, t}$. Recall that we consider value maximization in second-price auctions. If bidders underbid ($b_{i, t} < v_{i, t}$), they will always obtain non-negative utility and satisfy the RoS constraint, but not necessarily maximize their values. % \citep{aggarwal_auto-bidding_2024}.
So we assume overbidding. 
% the bidders to overbid. 

We introduce a condition/assumption on
% make a reasonable assumption on
the value and outside bid distributions, $F$ and $D$.
Let $(x)_+ = \max\{x, 0\}$ denote the positive part of a number. Let $v_{(N)}, v_{(N-1)}$ be the largest and second-largest values among $N$ i.i.d.~samples from $F$.
% The outside bid $d^O \sim D$.
% We assume that all values $v_i$ and outside bid $d^O$ are bounded by $B > 0$:
% \begin{align}
%     0\; \le \; v_i,\, d^O \; \le \; B. 
% \end{align}
% We further assume the following condition:
% We assume that:  
\begin{assumption} \label{ass:highest-value-assumption}
In expectation, the advantage of the second-largest value $v_{(N-1)}$ over the outside bid $d^O$ is larger than the advantage of $d^O$ over the largest value $v_{(N)}$: 
\[\Delta:=\E_{F, D}\big[ (v_{(N-1)} - d^O)_+ - ( d^O - v_{(N)})_+ \big] \ge 0.\] 
\end{assumption}

Intuitively, Assumption \ref{ass:highest-value-assumption} says that the largest two values among the coalition bidders are competitive enough against the outside bid.
% With large $v_{(N-1)}$ and $ v_{(N)}$, $\E[(v_{(N-1)} - d^O)_+]$ is large while $\E[( d^O - v_{(N)})_+]$ is small, which satisfies the assumption.
This assumption is satisfied when the number $N$ of coalition bidders is large enough: 
\begin{observation}
\label{obs:large-N-assumption-holds}
For any full-support distributions $F$ and $D$, Assumption \ref{ass:highest-value-assumption} is satisfied when $N$ is large enough. 
\end{observation}

% \tao{To justify this assumption, add a result showing that Assumption \ref{ass:highest-value-assumption} is satisfied as long as the number $N$ of coalition bidders is large enough. 
% And provide the following example to show that Assumption \ref{ass:highest-value-assumption} holds even with a small $N$. 
% }\gyr{This is no longer necessary, but the following example may be preserved since we still have our experiment results.}

Assumption \ref{ass:highest-value-assumption} is also satisfied by some natural distributions $F$ and $D$ with small $N$: 
\begin{example}
The following examples satisfy Assump.~\ref{ass:highest-value-assumption}: 
\begin{itemize}[itemsep=0pt, leftmargin=15pt]
    \item $N=4,F=U[0,1],D=U[0,1]$, with $\Delta=1/6$.
% The empirical results are demonstrated in Figure~\ref{fig:I2}.
    \item $N=3, F=U[0,1], D=\mathrm{Beta}(3,2)$, with $\Delta=1/40$. % The empirical results are demonstrated in Figure~\ref{fig:I3}.
\end{itemize}
\end{example}
% \noindent We perform simulations on these examples in Section \ref{sec:experiments}. 

% \vspace{0.4em}

The main result of this section is the following Theorem \ref{thm:utility/RoS-constraint}, which shows that Assumption \ref{ass:highest-value-assumption} is the necessary and sufficient condition for coordination to reduce RoS constraint violation compared to independent bidding. 
When $\Delta \ge 0$, coordination % guarantees uniform per-bidder utility gain
improves the utility of every bidder in the coalition, for any auto-bidding algorithm that weakly overbids.
When $\Delta<0$, this guarantee fails: % in a strong sense:
coordination can be strictly worse than independent bidding for some overbidding algorithm.
% The edge case $\Delta=0$ is discussed % separately in Appendix~\ref{app:delta-zero}.

\begin{theorem}
\label{thm:utility/RoS-constraint}
If ~\cref{ass:highest-value-assumption} holds (i.e., $\Delta \ge 0$), then for any (weakly) overbidding auto-bidding algorithm $A$ and every bidder $i\in[N]$,
\[
\mathbb{E}\!\left[U_{i, T}^{C,A}-U_{i, T}^{I,A}\right]\ge \tfrac{T\Delta}{N} \ge 0.
\]
If Assumption~\ref{ass:highest-value-assumption} fails (i.e., $\Delta<0$), then there exists an overbidding algorithm $A$ such that, for every bidder $i\in [N]$, for sufficiently large $T$,
\[
\mathbb{E}\!\left[U_{i, T}^{C,A} - U_{i, T}^{I,A}\right] < 0.
% \quad\text{for every } i\in[N].
\]
\end{theorem}

\begin{remark}
{\em
We can also obtain high-probability results from Theorem \ref{thm:utility/RoS-constraint}. 
Because values and outside bids are bounded by $B$, % $\Delta_{i, t}$ is bounded by $[-B, B]$. So,
we can use Azuma's inequality to prove that: when $\Delta > 0$, for any overbidding algorithm $A$,  
\begin{equation*}
     U_{i, T}^{C, A} ~ \ge ~ U_{i, T}^{I, A}  + \tfrac{T\Delta}{2N} 
\end{equation*}
holds with probability at least $1 - \exp\big(-\tfrac{T \Delta^2}{32 B^2 N^2}\big)$. 
}
\end{remark}

% \begin{conjecture}
% Coordination reduces RoS violation for any auto-bidding algorithm if and only if Assumption \ref{ass:highest-value-assumption} holds. 
% In particular,
% \begin{itemize}
%     \item If $\Delta > 0$, then for any over-bidding algorithm, $U_i^{C, A} > U_i^{I, A}$ (with high probability or expectation). 
%     \item If $\Delta < 0$, then there exists an over-bidding algorithm, such that $U_i^{C, A} < U_i^{I, A}$.
% \end{itemize}
% \end{conjecture}

% \tao{The proof for the second part goes to the appendix.}

\subsection{Proof of Theorem \ref{thm:utility/RoS-constraint}}
The proof of the second (negative) direction is in Appendix~\ref{app:negative}. Below we prove the first (positive) direction.

Let $U_{i, T}^{\truth} = \sum_{t=1}^T u_{i, t}^\truth$ be the total utility of bidder $i$ when all bidders in the coalition bid their values truthfully: $b^\truth_{i, t} = v_{i, t}$. 
We will prove Theorem \ref{thm:utility/RoS-constraint} by proving $\E[U_{i, T}^{C, A}] \ge \E[U_{i, T}^{\truth}] + \frac{\Delta T}{N}$ and $U_{i, T}^{\truth} \ge U_{i, T}^{I,A}$.
The second inequality is formalized below:  
\begin{lemma} \label{lem:U-truthful-vs-U-I-A}
For any sequence of values and outside bids $(\bm v_t, d_t^O)_{t=1}^T$, 
$U_{i, T}^{\truth} \ge U_{i, T}^{I,A}$.
\end{lemma}
\begin{proof}
This lemma directly follows from the dominant-strategy-incentive-compatibility of second-price auction, which ensures that each bidder's truthful-bidding utility $u^\truth_{i, t}$ is weakly larger than the non-truthful bidding utility $u^{I, A}_{i, t}$ at every round. 
\end{proof}

We then prove $\E[U_{i, T}^{C, A}] \ge \E[U_{i, T}^{\truth}] + \frac{T \Delta}{N}$: % happens with high probability. 

\begin{lemma} \label{lem:U-C-A-vs-U-truthful}
Under Assumption \ref{ass:highest-value-assumption}, for any auto-bidding algorithm that overbids, for every bidder $i\in[N]$, 
% with probability at least $1-\exp\big(-\frac{\Delta^2 T}{32 B^2 N^2}\big)$, 
    % $U_i^{C, A} - U_i^{\truth} = \sum_{t=1}^T \InParentheses{u_{i,t}^{C,A} - u_{i,t}^{\truth}} \ge \frac{\Delta T}{2N}$.
\begin{equation*}
\E\left[U_{i, T}^{C, A} - U_{i, T}^{\truth}\right]=\sum \nolimits_{t=1}^T \InParentheses{u_{i,t}^{C,A} - u_{i,t}^{\truth}}\ge \tfrac{T \Delta}{N}. 
\end{equation*}
\end{lemma}

\begin{proof}
See Appendix~\ref{app:proof-of-lemU-C-A-vs-U-truthful}.
\end{proof}

Lemmas \ref{lem:U-truthful-vs-U-I-A} and \ref{lem:U-C-A-vs-U-truthful} together prove Theorem \ref{thm:utility/RoS-constraint}. 

% =======

% Non i.i.d case.

% \begin{align*}
%     \E[\Delta_{i,t}|H_{t-1}]\ge & \E[\1[v_i\text{is max}](\max_{j\neq i} v_j-d_t^O)_+]\\
%     &-\E[\1[v_i\text{is max}](d_t^O-v_i)_+]]
% \end{align*}

% \begin{align*}
%     \sum_{t=1}^T\E[\Delta_{i,t}|H_{t-1}]\ge \E[(v_{(N-1)}-d_t^O)_+]-\E[(d_t^O-v_{(N)})_+]\ge \Delta
% \end{align*}

\section{Coordination Increases Bidders' Value}\label{sec:value}
In this section, we show that coordination can increase the total value of auto-bidders, compared to independent bidding. This holds for a broad class of auto-bidding algorithms and without any distributional assumptions. 

In particular, we consider any \emph{mirror-descent-based} auto-bidding algorithm.
They are state-of-the-art algorithms for value maximization under RoS constraint, with near-optimal $O(\sqrt{T})$ regret guarantees \citep{balseiro_best_2023, feng2023online}.
The algorithm maintains a parameter $\lambda_{i, t}$ for each bidder (which corresponds to the Lagrange multiplier of the dual optimization problem), and the bidder overbids by $b_{i, t} = (1+\frac{1}{\lambda_{i, t}}) v_{i, t}$ at each round. 
After observing a utility feedback, the bidder adjusts the parameter $\lambda_{i, t}$ in the opposite direction of the utility (i.e., positive utility decreases $\lambda_{i, t}$). 
The full algorithm is given in Algorithm \ref{alg:general-mirror-descent}. 

\begin{definition}[Legendre mirror map and Bregman divergence]
A $C^1$ function $h:(0,\infty)\to\R$ is a \emph{Legendre mirror map} if:
(i) $h$ is strictly convex on $(0,\infty)$;
(ii) $h'$ is a bijection from $(0,\infty)$ onto $\R$
(equivalently, $h'$ is strictly increasing with
$\lim_{\lambda\downarrow 0}h'(\lambda)=-\infty$ and $\lim_{\lambda\uparrow\infty}h'(\lambda)=+\infty$).
The associated Bregman divergence is
\[
D_h(\lambda,\mu)
:= h(\lambda)-h(\mu)-h'(\mu)(\lambda-\mu),
\quad \forall \lambda,\mu>0.
\]
\end{definition}

\begin{algorithm}[H]
\caption{Mirror-Descent RoS Auto-Bidder (MD-$h$)}
\label{alg:general-mirror-descent}
    \begin{algorithmic}
       \STATE Initialize multiplier $\lambda_{i, 1} = 1$, learning rate $\alpha = 1/\sqrt{T}$. 
       \FOR{$t=1,2,\cdots,T$}
       \STATE Observe value $v_{i,t}$ and set bid $b_{i,t}=(1 + \frac{1}{\lambda_{i,t}}) v_{i,t}$.
       \STATE Obtain allocation $x_{i,t}$ and pay $p_{i,t}$, compute realized utility $g_{i,t}=v_{i,t}\cdot x_{i,t}-p_{i, t}$.
       \STATE Update the multiplier: $$\lambda_{i,t+1} = \arg\min_{\lambda>0}\Big\{ \alpha \cdot g_{i,t}\cdot \lambda + D_h(\lambda,\lambda_{i,t})\Big\}. $$
       \ENDFOR
    \end{algorithmic}
\end{algorithm}

% \begin{definition}[Mirror-Descent RoS Auto-bidder (MD-$h$)]
% Fix a Legendre mirror map $h$ and set step size $\alpha_T:=1/\sqrt{T}$.
% For each bidder $i\in[N]$, the MD-$h$ auto-bidder maintains multipliers $\lambda_{i,t}\in(0,\infty)$ with $\lambda_{i,1}=1$.
% At each round $t$:
% \begin{enumerate}
% \item observe value $v_{i,t}$ and bid
% \[
% b_{i,t} := \Bigl(1+\frac{1}{\lambda_{i,t}}\Bigr)v_{i,t};
% \]
% \item observe $(x_{i,t},p_{i,t})$ and form utility feedback $g_{i,t}=v_{i,t}x_{i,t}-p_{i,t}$;
% \item update the multiplier by mirror descent:
% \begin{equation}
% \lambda_{i,t+1}\in\arg\min_{\lambda>0}\Bigl\{\alpha_T g_{i,t}\cdot \lambda + D_\phi(\lambda,\lambda_{i,t})\Bigr\}.
% \tag{MD-$h$}
% \end{equation}
% \end{enumerate}
% \end{definition}

As an example, \citet{feng2023online} study the mirror-descent-based algorithm with the \emph{entropy mirror map} $h(\lambda)=\lambda\log\lambda-\lambda$, in which case the mirror-descent update is
\begin{equation} \label{eq:feng-update-rule} \lambda_{i,t+1}=\lambda_{i,t}\exp(-\alpha g_{i,t}).
\end{equation}
\citet{feng2023online}'s algorithm guarantees that each bidder's total RoS violation in $T$ rounds is at most $O(\sqrt{T} \log T)$ and total value is at most $O(\sqrt{T})$ away from the optimum. 

Let
\[V_{t}^{C, A} = \sum \nolimits_{i = 1}^N  v_{i, t} x^{C, A}_{i, t} \qquad V_{t}^{I, A} = \sum \nolimits_{i = 1}^N  v_{i, t} x^{I, A}_{i, t}\]
be the total values obtained by the $N$ bidders at round $t$ in the coordinated and independent cases, respectively.
We show that, for any mirror-descent-based algorithm $A$, coordination always improves the total value of coalition bidders. 

\begin{theorem}
\label{thm:value-wo-ass}
Suppose bidders run a mirror-descent algorithm A (Algorithm \ref{alg:general-mirror-descent}).
As $T \to \infty$, the coalition bidders' total value under coordinated bidding is weakly larger than that under independent bidding: 
\[
\lim_{T\to\infty}\E\Big[ \tfrac{1}{T}\sum \nolimits_{t=1}^T V_t^{C, A} \;-\; \tfrac{1}{T}\sum \nolimits_{t=1}^T V_t^{I, A}\Big] \ge 0.
\]    
\end{theorem}

\begin{proof}[\bf Proof sketch.]
We analyze the mirror-descent dynamics in the derivative space
$y_{i,t}:=h'(\lambda_{i,t}) \in (-\infty,+\infty)$, where the mirror-descent update becomes $y_{i,t+1} = y_{i,t} - \alpha g_{i,t}$.
Let $G_{(N)}(\lambda)$ be the expected utility of a \emph{single} bidder whose value is distributed according to the highest value $v_{(N)} \sim F_{(N)}$ among $N$ bidders, and who bids using multiplier $\lambda > 0$, competing against the outside bid $d^O \sim D$:  
\begin{align*}
    G_{(N)}(\lambda) := \E \Big[(v_{(N)} - d^O) \cdot \1\big[ (1+\tfrac{1}{\lambda}) v_{(N)} > d^O\big]\Big]. 
\end{align*}
Let $V_{(N)}(\lambda) := \E \big[ v_{(N)} \cdot \1[ (1+\tfrac{1}{\lambda}) v_{(N)} > d^O] \big]$ denote the expected value of this single bidder.

\begin{lemma}
$G_{(N)}(\lambda)$ is strictly increasing in $\lambda > 0$.  
\end{lemma}
\begin{proof}
    See Lemma~\ref{lemma:mono}.
\end{proof}

% Let $\lambda_\star := \inf\{\lambda>0 : G(\lambda)\ge 0\}, y_\star := h'(\lambda_\star).$
Let $\lambda_\star := \inf\{\lambda>0:G_{(N)}(\lambda)\ge0\}$. When $\lambda_\star=0$, denote $V_{(N)}(0) := \lim_{\lambda\downarrow0}V_{(N)}(\lambda)$. When $\lambda_\star>0$, let $y_\star := h'(\lambda_\star)$. 
% The proof proceeds by identifying $\lambda_\star$ as a stable reference point for the coordinated mirror-descent dynamics
The proof proceeds by proving that the multipliers in the coordinated case converge to $\lambda_\star$ eventually
and then upper bounding independent bidding through a single-bidder relaxation. 
% relating stability in $\lambda$ to time-average value.

We first characterize the expected update direction under coordination. Let $H_{t-1} := \bigcup_{i=1}^N H_{i,t-1}$ be the joint history up to round $t-1$.

\begin{claim}
\label{clm:coord-drift}
Under coordinated bidding, conditioning on history $H_{t-1}$,
the expected realized utility (which serves as the gradient direction for multiplier update) satisfies
\[
\E\big[ g_{i,t}\mid H_{t-1} \big] = \tfrac{1}{N} G_{(N)}(\lambda_{i,t}).
\]
\end{claim}

When $\lambda_\star>0$, since $G_{(N)}(\lambda)$ is increasing, we have $G_{(N)}(\lambda_{i, t})\gtrless 0$ if and only if $\lambda_{i, t} \gtrless \lambda_\star$, giving expected gradient direction $\E[g_{i,t} \,|\, H_{t-1}] \gtrless 0$, so the mirror-coordinate update pushes $y_{i,t}$ toward $y_\star$. 
This property allows us to view $(y_{i,t}-y_\star)^2$ as a potential function, which decreases on average whenever $\lambda_{i,t}$ deviates
significantly from $\lambda_\star$. The boundary case $\lambda_\star=0$ is handled by a one-sided exponential potential. These arguments imply: 
\begin{claim}
\label{clm:coord-stability}
For any $\delta>0$, the fraction of rounds $t\le T$ for which the active bidder $i_t^\star = \arg\max_{i}v_{i,t}$ is away from the reference multiplier $\lambda_\star$ vanishes in expectation as $T\to\infty$.
\end{claim}

% The above convergence property implies that, asymptotically, almost all value contribution under coordination comes from rounds in which the effective multiplier is close to $\lambda_\star$.

The above claim implies that, asymptotically, the active bidder's multiplier in coordinated bidding converges to $\lambda_\star$. % (or $0$ in the boundary case).
As a result, the coalition bidders' total value converges to the expected value obtained by a single bidder whose value is $v_{(N)} \sim F_{(N)}$ and who bids by multiplier $\lambda_\star$: 
% almost all value contribution under coordination comes from rounds in which the effective multiplier is close to $\lambda_\star$.

\begin{claim}
\label{clm:coord-value}
The coordinated time-average total value converges to the single bidder value
$$
\lim_{T\to\infty} \E\Big[\tfrac 1T \sum \nolimits_{t=1}^T V_t^{C,A}\Big] ~ = ~ V_{(N)}(\lambda_\star).
$$
\end{claim}
We then upper bound the time-average total value under independent bidding.

\begin{claim}
\label{clm:ind-upper}
Under independent bidding,
\[
\limsup_{T\to\infty} \E\Big[\tfrac{1}{T}\sum \nolimits_{t=1}^T V_t^{I,A}\Big] ~ \le ~ V_{(N)}(\lambda_\star).
\]
\end{claim}

To see this, view $\bar b_t=\max_i b_{i,t}$ as the bid of a \emph{virtual single bidder} with value $v_{(N),t}$; its attained value upper bounds the total value attained by the independent coalition. After clipping each bidder's utility below by $-B$, the utility of the virtual bidder dominates the sum of clipped utilities of the coalition bidders, and the MD updates imply a deterministic $o(T)$ RoS violation bound for the virtual bidder. A Lagrangian envelope for the single-bidder problem then upper bounds its average value by $V_{(N)}(\lambda)+\lambda G_{(N)}(\lambda)$ for every $\lambda>0$, whose infimum is $V_{(N)}(\lambda_\star)$.

Combining Claim~\ref{clm:ind-upper} with Claim~\ref{clm:coord-value} proves the theorem. 
Details are in Appendix \ref{proof:value-wo-ass}. 
\end{proof}

Theorem \ref{thm:value-wo-ass} does not require any assumption on the distributions of values or outside bid, except for boundedness and continuity. We further show that, if Assumption \ref{ass:highest-value-assumption} holds, then coordinated mirror descent is not only better than independent mirror descent but also weakly better than \emph{any other coordination mechanism}. 

\begin{definition}[Coordination Mechanism]
\label{def:coord-alg}
A \emph{coordination mechanism} $\mathcal G$ is any procedure that, at each round $t$, given all realized values
$\bm v_t = (v_{1,t}, \ldots, v_{N,t})$, selects bids $(b_{1,t},\ldots,b_{N,t})$ for the coalition bidders.
The mechanism may be history-dependent, i.e., bids at round $t$
may depend arbitrarily on past observations. % and actions.
\end{definition}

\begin{theorem}
\label{thm:value}
Under Assumption \ref{ass:highest-value-assumption},
as $T \to \infty$, %in expectation,
any coordinated mirror-descent algorithm $A$ is optimal in maximizing coalition bidders' total value.
% the coalition bidders' total value under coordinated mirror-descent algorithm $A$ (Algorithm \ref{alg:general-mirror-descent}) is weakly larger than that under
Formally, for any coordination mechanism $\mathcal G$, 
\[
\lim_{T\to\infty}\E\Big[ \tfrac{1}{T}\sum \nolimits_{t=1}^T V_t^{C, A} \;-\; \tfrac{1}{T}\sum \nolimits_{t=1}^T V_t^{\mathcal G}\Big] \ge 0,
\]
where $V_t^{\mathcal G}$ is the total value obtained by the $N$ bidders at round $t$ using mechanism $\mathcal G$.
\end{theorem}

% \tao{I changed ``coordination algorithm $\mathcal A$'' to ``coordination mechanism $\mathcal G$''}

% Algorithm~\ref{alg:feng-algorithm} corresponds to the choice $\phi(\lambda)=\lambda\log\lambda-\lambda$. With a fixed Legendre mirror-map $h$, we formalize the associated mirror-descent auto-bidding algorithm as follows.

\begin{proof}[\bf Proof sketch.]
To prove Theorem \ref{thm:value}, we first reduce the $N$-bidder coordinated scenario to a single-bidder problem where the single bidder's value is the maximum among the coalition members. Using this reduction, we prove an interesting fact: under Assumption \ref{ass:highest-value-assumption}, the multipliers $\lambda_{i, t}^C$ of coordinated bidders converge to $0$ in the limit as $t\to\infty$, meaning that the coordinated bidders will eventually submit arbitrarily high bids. This implies that the coordinated coalition will obtain the highest possible value, among all coordination mechanisms. The full proof is presented in Appendix \ref{app:proof-of-value}.
\end{proof}

\section{Extension to Non-i.i.d. Valuations}
\label{sec:asymmetric}

% \gyr{Why do we extend to non i.i.d.}

% We relax the i.i.d.\ assumption and allow bidder-specific value distributions.

This section discusses the extension where bidders have independent but non-identically distributed values. 
In each round $t \in [T]$, bidders $i \in [N]$ draw values $v_{i,t}$ from different $F_i$. %(independently across bidders and rounds),
The outside bid $d^O_t$ is still drawn from $D$.

% Values and outside bids are bounded by $B>0$. 
% Let $V_t^{C}$ and $V_t^{I}$ denote the coalition's total value achieved in round $t$ under \emph{coordinated} and \emph{independent} bidding, respectively. 
% Let $U_i^{\cdot, A}$ denote the cumulative RoS violation (achieved value minus payment) of bidder $i$ over $T$ rounds under scenario ${\cdot}\in\{C,I\}$ when all bidders follow auto-bidding algorithm $A$.

\paragraph{Total RoS.}
In the i.i.d.~case, coordination reduces the RoS constraint violation for \emph{every} bidder in the coalition, if and only if Assumption \ref{ass:highest-value-assumption} holds. 
In the non-i.i.d.~case, we show that under the same assumption, coordination reduces the \emph{total} RoS constraint violation among the coalition bidders, namely, the total utility of the coalition is improved. 
Recall that $ U_{i, T}^{C, A} = \sum \nolimits_{t=1}^T u_{i, t}^{C, A}$ and $U_{i, T}^{I, A} = \sum \nolimits_{t=1}^T u_{i, t}^{I, A}$ are the total utilities of bidder $i$ during $T$ rounds under coordinated and independent bidding scenarios. 
% Our first result shows that, in the non-i.i.d.~setting, coordination strictly reduces the \emph{total} RoS constraint violation, namely, the total utility of the coalition bidders is improved: 
% \gyr{Difference with iid}
% even when bidders' values are not i.i.d.
% \gyr{Change to expectation version.}
\begin{theorem}[Coordination reduces total RoS violation, non-i.i.d. case]
%under Heterogeneity
\label{thm:asym-violation}
Under Assumption~\ref{ass:highest-value-assumption} (i.e., $\Delta \ge 0$), for any overbidding algorithm $A$,
% with probability at least $1 - \exp\!\big(-\tfrac{\Delta^2 T}{32B^2}\big)$, the total utility $\sum_{i=1}^N U_i^{C, A} \ge \sum_{i=1}^N U_i^{I, A} + \frac{\Delta T}{2}$. 
the total utility among the $N$ coalition bidders satisfies 
\[\E\Big[\sum \nolimits_{i=1}^N U_{i, T}^{C, A}\Big] ~ \ge ~ \E\Big[\sum \nolimits_{i=1}^N U_{i, T}^{I, A}\Big] + \Delta T.\] 
\end{theorem}
% \begin{proof}
The proof of this theorem is in Appendix~\ref{app:thm-asym-vio}. 
% \end{proof} 

Although Theorem \ref{thm:asym-violation} does not guarantee RoS improvement for every bidder, the total RoS improvement is also desirable in practice because the coalition bidders may redistribute the extra utility among themselves.   

\paragraph{Total value.}
In the i.i.d.~setting, we showed that coordinated mirror descent algorithms achieve higher total value than independent mirror descent algorithms without any assumption, and higher than any coordination mechanism under Assumption \ref{ass:highest-value-assumption}. 
The proof, which relied on the fact that each bidder has the highest value among the coalition with probability $1/N$,  
% In the non-i.i.d.\ setting, however, comparing %total value requires additional structure,
% coordinated and independent bidding requires additional structure, 
% since bidders may differ in how often they have the highest value. %are the highest-value coalition member.
% The proof
does not apply to the non-i.i.d.~case because bidders now differ in how often they are the highest-value coalition member. 
We therefore impose a per-bidder condition ensuring that each bidder contributes positively whenever they are selected.
Concretely, for bidder $i$, define
\[
\Delta_i := \E\big[(v_i - d^O)\,\1\{v_i = \max_{j \in [N]} v_j\}\big].
\]

% To compare total value, we require a per-bidder margin condition that ensures each bidder wins with a positive average advantage when they are the highest-value coalition member. 

\begin{assumption}[Positive per-bidder winning margin]
\label{ass:non-iid-value}
% There exist constants $C_1,\ldots,C_N>0$ such that $\Delta_i \ge C_i$ for every $i \in [N]$.
For each bidder $i \in [N]$, $\Delta_i > 0$.
\end{assumption}

\begin{theorem}
% [Total Value under non-i.i.d. valuations]
[Coordination increases total value, non-i.i.d. case]
\label{thm:asym-value}
Under Assumption \ref{ass:non-iid-value}, coordinated mirror descent algorithms
% suppose bidders run
(\Cref{alg:general-mirror-descent}) achieve weakly higher total value than independent mirror descent algorithms and any other coordination mechanism $\mathcal G$: 
% Then
% \[\lim_{T\to\infty} \E \Big[\frac{1}{T}\sum_{t=1}^T V^{C, A}_t - \frac{1}{T}\sum_{t=1}^T V^{I, A}_t\Big] \ge 0.\]
\[\lim_{T\to\infty} \E \Big[\tfrac{1}{T}\sum \nolimits_{t=1}^T V^{C, A}_t - \tfrac{1}{T}\sum \nolimits_{t=1}^T V^{\mathcal G}_t\Big] \ge 0.\]
\end{theorem}

% \begin{proof}
See Appendix~\ref{app:thm-asym-value} for the proof.
% \end{proof}

\vspace{0.4em}
The findings in this section suggest that coordination improves upon independent bidding even in the presence of asymmetry inside the coalition, demonstrating the robustness of our conclusions.  

\begin{table*}[t]
\centering
\small
\caption{Performance summary. Confidence intervals no greater than $\pm 0.0025$ are omitted. All totals are normalized by $T$.}
\label{tab:summary-table}
\begin{tabular}{cccccccc}
\hline
Figure  & Setting & $N$ & $T$ & Total Utility (I) & Total Utility (C) & Total Value (I) & Total Value (C) \\ 
\hline
\ref{fig:I1} & i.i.d. & 2 & 4000 & -0.011 & 0.220 & 0.643 & 0.666 \\
\ref{fig:I2} & i.i.d. & 4 & 4000 & -0.077  & 0.302 & 0.774 & 0.800 \\
\ref{fig:I3} & i.i.d. & 3 & 4000 & -0.049  & 0.153 & 0.712 & 0.748 \\
\ref{fig:NI1}  & non-i.i.d. & 2 & 10000 & 0.049 & 0.219& 0.715 & 0.718 \\
\ref{fig:NI2}  & non-i.i.d. & 3 & 20000 & -0.014 & 0.258 & 0.619 & 0.633 \\
\ref{fig:NI3} & non-i.i.d. & 5 & 20000 & -0.062 & 0.619 & 0.814 & 0.819 \\
\ref{fig:REAL-MERGE-K4} & Real data & 4 & 20000 & -0.040 & 0.155 $\pm$ 0.012 & 0.620 $\pm$ 0.016 & 0.928 $\pm$ 0.003 \\
\ref{fig:REAL-MERGE-K5} & Real data & 5 & 20000 & -0.065 & 0.172 $\pm$ 0.012 & 0.608 $\pm$ 0.012 & 0.958\\
\hline
\end{tabular}
\end{table*}

% \begin{figure*}[t]
%   \centering

%   % ---------- first row: two figures ----------
%   \begin{subfigure}[b]{0.48\textwidth}
%     \centering
%     \includegraphics[width=\textwidth]{pictures/I1.png}
%     \caption{$N=2,F=U[0,1], D=U[0,0.9]$}
%     \label{fig:I1}
%   \end{subfigure}
%   \hfill
%   \begin{subfigure}[b]{0.48\textwidth}
%     \centering
%     \includegraphics[width=\textwidth]{pictures/I2.png}
%     \caption{$N=4,F=U[0,1], D=U[0,1]$}
%     \label{fig:I2}
%   \end{subfigure}

%   \medskip

%   % ---------- second row: one centered ----------
%   \begin{subfigure}[b]{0.48\textwidth}
%     \centering
%     \includegraphics[width=\textwidth]{pictures/I3.png}
%     \caption{$N=3,F=U[0,1], D=\mathrm{Beta}(3,2)$}
%     \label{fig:I3}
%   \end{subfigure}

%   \caption{Experiments under i.i.d.\ auto-bidders.}
%   \label{fig:iid}
% \end{figure*}

\section{Experiments}
\label{sec:experiments}
To complement our theoretical analysis, we conduct a series of experiments evaluating the performance of coordinated auto-bidding, % We consider both controlled synthetic environments and the public iPinYou dataset \citep{ZhangYuanWang2014iPinYou,LiaoEtAl2014iPinYou}, using them to validate our theoretical results and to stress-test robustness under milder assumptions.
on both synthetic data and a public real-world dataset (iPinYou \citep{ZhangYuanWang2014iPinYou,LiaoEtAl2014iPinYou}).

Unless otherwise stated, all metrics are averaged over 100 independent simulations; shaded regions in % Figures~\ref{fig:iid}, \ref{fig:non-iid}, \ref{fig:real-data} (See Appendix~\ref{app:fig})
figures depict 95\% confidence intervals computed as $\bar{x} \pm 1.96 \cdot s$, where $\bar{x}$ and $s$ are the sample mean and sample standard deviation divided by $\sqrt{100}$. % over the 100 simulations.

\subsection{Experiment Setup}
In all experiments, we run Algorithm~\ref{alg:general-mirror-descent} with mirror map
$h(\lambda)=\lambda \log \lambda - \lambda$ \cite{feng2023online} (i.e., \Cref{eq:feng-update-rule}) and learning rate
$\alpha = 1/\sqrt{T}$, where $T$ is the total number of auction rounds. The horizon $T$ may vary across experiments because different settings exhibit different convergence rates. The values of $T$ and $N$ (number of coalition bidders) used in each setting are summarized in Table~\ref{tab:summary-table}.

\subsubsection{Synthetic Data}

\paragraph{Symmetric setting.}

Each coalition bidder $i$ draws values i.i.d.\ from uniform distribution $U[0,1]$. %, matching the stylized cases in Figure~\ref{fig:iid}. 
The outside bid is sampled from $D=U[0,0.9]$ (Figure \ref{fig:I1}), $D=U[0,1]$ (Figure \ref{fig:I2}), or $D=\mathrm{Beta}(3,2)$ (Figure~\ref{fig:I3}), representing different market conditions faced by the coalition bidders. 
% thus, only the market distribution varies across the symmetric tests.

\paragraph{Asymmetric setting.}
To test the coordination mechanism under heterogeneity, we sample each bidder’s values from distinct families (Uniform, Beta, or truncated Gaussians with bidder-specific parameters). The outside-bid distribution $D$ varies with the scenario: $U[0.2,0.8]$ in Figure~\ref{fig:NI1}, $\mathrm{Beta}(3,5)$ in Figure~\ref{fig:NI2}, and $\mathrm{Beta}(2,8)$ in Figure~\ref{fig:NI3}.
Asymmetries arise from both inside the coalition and the outside market conditions.

Full distributional specifications for the synthetic experiments are provided in the figure captions (Figures \ref{fig:I1}, \ref{fig:I2}, \ref{fig:I3}, \ref{fig:NI1}, \ref{fig:NI2}, \ref{fig:NI3}). All distributions conform to Assumption \ref{ass:highest-value-assumption}.

\subsubsection{Real-World Data}

% We use the public iPinYou dataset \citep{ZhangYuanWang2014iPinYou,LiaoEtAl2014iPinYou}, released for a real-time bidding competition organized by iPinYou.

% We use the Season~2 test split, which contains $2{,}521{,}630$ auction records from $5$ advertisers. Each record consists of a bidding (winning) price and a corresponding advertiser ID.

% \textcolor{red}{Due to the limitation of the real data set: (1) we only get the winning bid (2) cannot access the identity of the advertisers (3) has no information about whether any bidders are coordinated or not, we simulate the value of coordinated bidders through drawing a bid uniformly at random with replacement, which is equivalent to the i.i.d setting with value distribution from the real data set.}

% Following prior work~\citep{chen_coordinated_2023}, we first linearly normalize all bidding prices to the interval $[0,1]$ to ensure numerical stability. Subsequently, all bids are aggregated into a common pool, from which each bidder independently draws a bid uniformly at random with replacement. The outside bidder’s price is sampled from the same distribution but scaled by a random factor drawn from $U[1,2]$, ensuring that its bids remain within a competitive range. We conduct experiments with $N = 4$ and $N = 5$ bidders under this setup.

% ===

We also experimented with the public iPinYou dataset \citep{ZhangYuanWang2014iPinYou,LiaoEtAl2014iPinYou}, released for a real-time bidding competition organized by iPinYou. We use the Season~2 test split, which contains $2{,}521{,}630$ auction records from $5$ advertisers. Each record provides the bidding (winning) price and the corresponding advertiser ID.

The real-world data impose three limitations: (i) only the winning price is observed, (ii) bidder identities are unavailable, and (iii) there is no label indicating whether any bidders coordinate. To obtain bidder valuations consistent with the data while abstracting away unobserved identities and strategies, we adopt an empirical i.i.d.\ model: we aggregate all observed winning prices into a common pool and, in each round, draw each coalition bidder’s value independently and uniformly at random with replacement from this pool. This is equivalent to sampling i.i.d.\ values from the empirical price distribution induced by the dataset. 
Constructing a reliable non-i.i.d.\ model would require persistent bidder identities
or side information about coordination patterns, which are not available in the data.

%Following prior work \citep{chen_coordinated_2023},
We normalize all prices to $[0,1]$ for numerical stability. % Under this setup, 
The outside bidder’s bid is drawn from the same empirical distribution and then scaled by a random factor from $U[1,2]$ to maintain competitive pressure. We report results for coalitions with $N\in\{4,5\}$ bidders under this setup.

\subsection{Results}

\paragraph{Metrics.}
Across all experiments, we evaluate two metrics for both \textsc{Independent} (Algorithm~\ref{alg:independent-bidding}) and \textsc{Coordinated} (Algorithm~\ref{alg:coordinated-bidding}) algorithms: (a) \emph{utility} (achieved value minus payment), (b) \emph{achieved value}. 
%, and (c) \emph{multiplier}.
Higher utility indicates lower RoS violation; in particular, utility $\ge 0$ implies that the RoS constraint is satisfied. Figures report per-bidder trajectories, whereas Table~\ref{tab:summary-table} aggregates to coalition-level totals (sum across bidders). All statistics are averaged over 100 independent runs.

\paragraph{Plot notation.}
% To distinguish conditions within the plots,
In the figures, labels of the form \texttt{bidder $i$ - I} denote bidder $i$ under the Independent scenario, whereas \texttt{bidder $i$ - C} denotes the Coordinated scenario. Shaded regions show 95\% confidence intervals over 100 simulations.

\paragraph{Results on synthetic data}
In the symmetric settings (Figures~\ref{fig:I1}, \ref{fig:I2}, \ref{fig:I3}), coordination sharply increases both each bidder’s utility (hence reducing RoS violation) and the coalition’s total achieved value relative to independent bidding. Although our theory only guarantees value improvement at the coalition (total) level, the experiments further show that \emph{every individual bidder} benefits from coordination by obtaining higher value. 
% in the symmetric case.
% Moreover, all multipliers converge to $0$, consistent with our theory, and
We observe that per-bidder curves %(violation, value, and multiplier trajectories)
closely track one another due to the i.i.d.\ value draws.

% \begin{figure}[htbp]
%   \centering
%   \includegraphics[width=\linewidth]{}
%   \caption{Experiments under i.i.d.\ auto-bidders
%   ($N=2, F=U[0,1], D=U[0,0.9]$).}
%   \label{fig:I1}
% \end{figure}

In the asymmetric settings (Figures~\ref{fig:NI1}, \ref{fig:NI2}, \ref{fig:NI3}), the aggregate gains largely persist: coordination improves total value on average and typically improves violations. However, heterogeneity induces dispersion in learning dynamics -- bidders with different value distributions exhibit different convergence rates, %of the multiplier,
and individual value improvements are not guaranteed for every bidder even though the coalition-level total value increases.
% In cases where each bidder has a non-negligible chance to be the winning bidder, at least occasionally, we also observe that the multipliers tend to converge toward $0$ empirically.

\begin{figure}[htbp]
  \centering
  \begin{subfigure}[b]{0.48\textwidth}
    \centering
    \includegraphics[width=\linewidth]{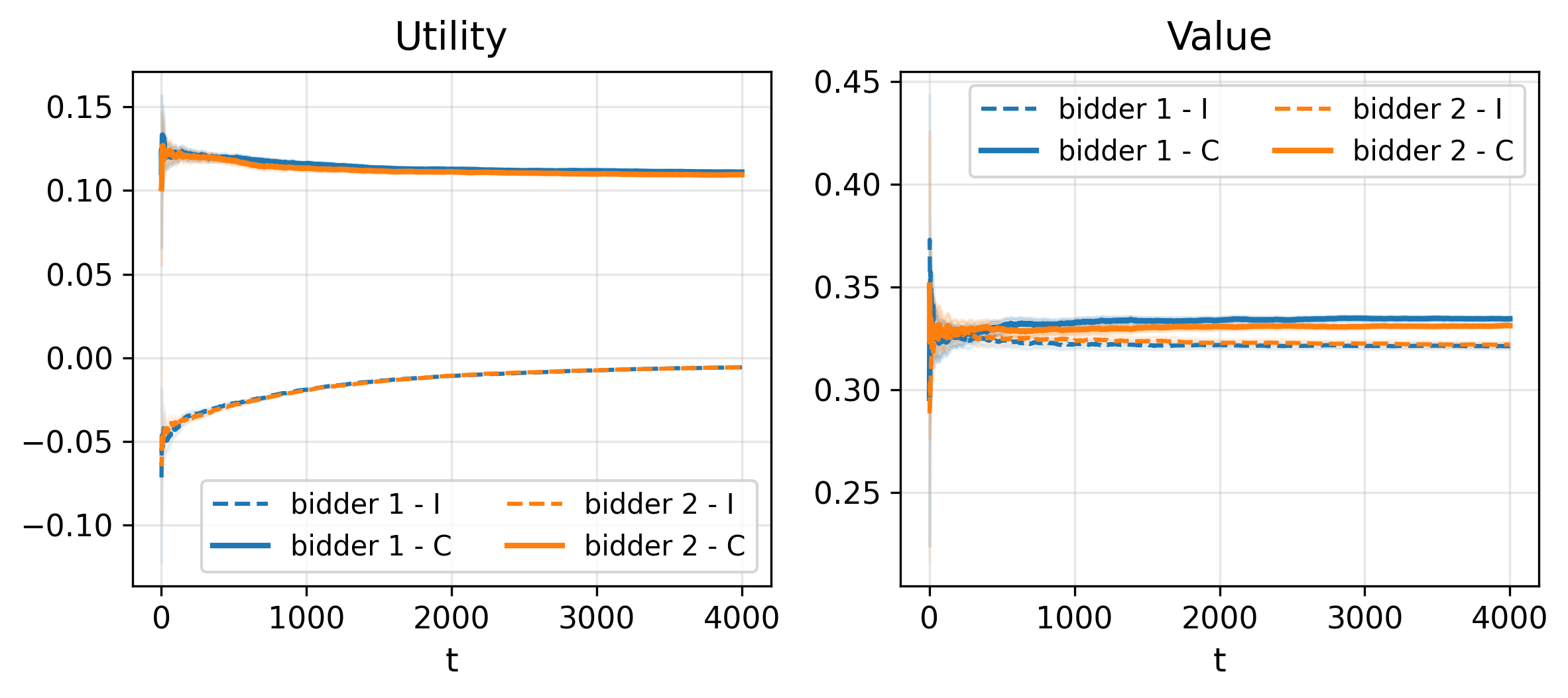}
  \caption{Experiments under i.i.d.\ auto-bidders
  ($N=2, F=U[0,1], D=U[0,0.9]$).}
  \label{fig:I1}
  \end{subfigure}
  \hfill
  \begin{subfigure}[b]{0.48\textwidth}
    \centering
    \includegraphics[width=\linewidth]{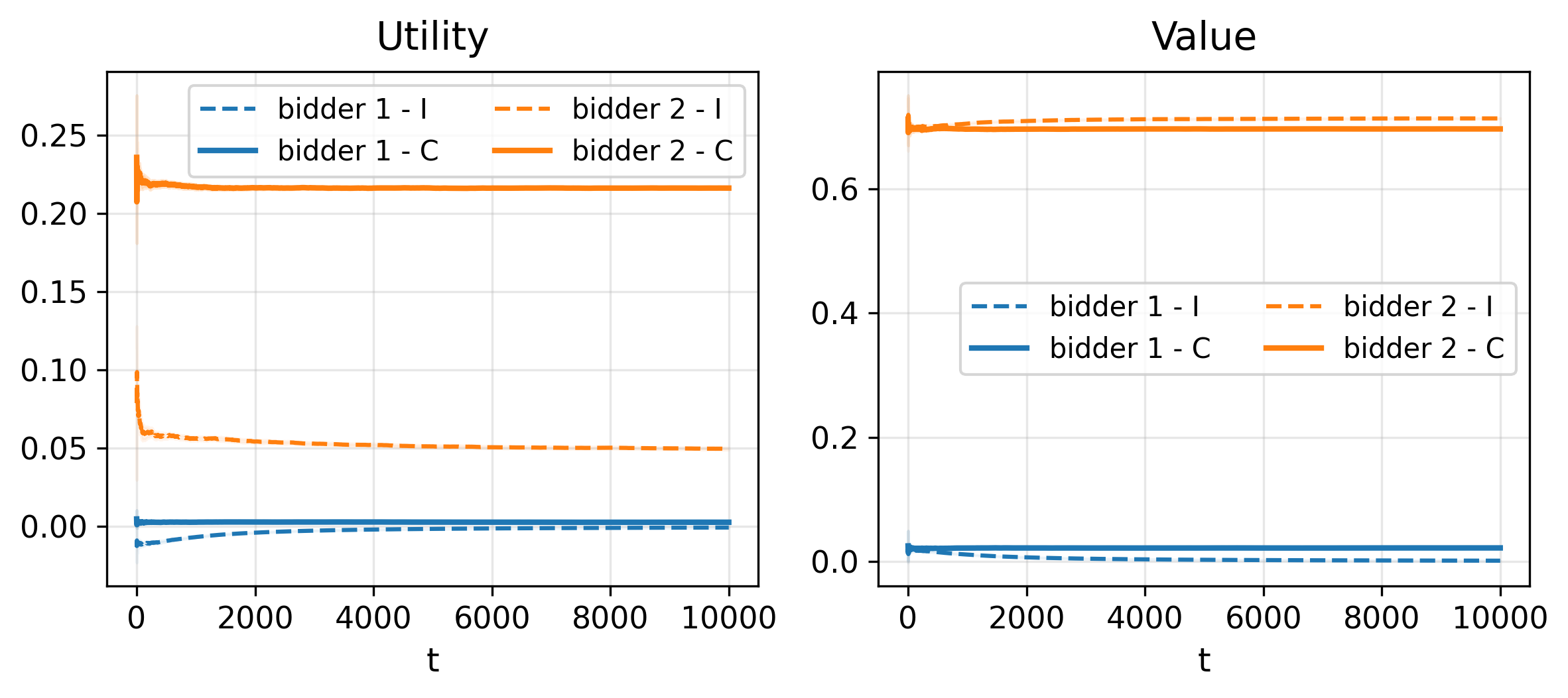}
      \caption{Experiments under non-i.i.d.\ auto-bidders
      ($N=2,F_1=\mathrm{Beta}(2,5),F_2=\mathrm{Beta}(5,2), D=U[0.2,0.8]$).}
      \label{fig:NI1}
  \end{subfigure}
  \caption{Experiments on synthetic data.}
\end{figure}

% \begin{figure}[htbp]
%   \centering
%   \includegraphics[width=0.48\linewidth]{pictures/NI1.png}
%   \caption{Experiments under non-i.i.d.\ auto-bidders
%   ($N=2,F_1=\mathrm{Beta}(2,5),F_2=\mathrm{Beta}(5,2), D=U[0.2,0.8]$).}
%   \label{fig:NI1}
% \end{figure}

\paragraph{Results on real-world data.} Consistent with the synthetic data experiments, real-world data results (Figures \ref{fig:REAL-MERGE-K4}, \ref{fig:REAL-MERGE-K5}) show improvements in both individual RoS violation and achieved value.
% , accompanied by multipliers converging steadily toward~$0$.

\begin{figure}[htbp]
  \centering
  \includegraphics[width=0.48\linewidth]{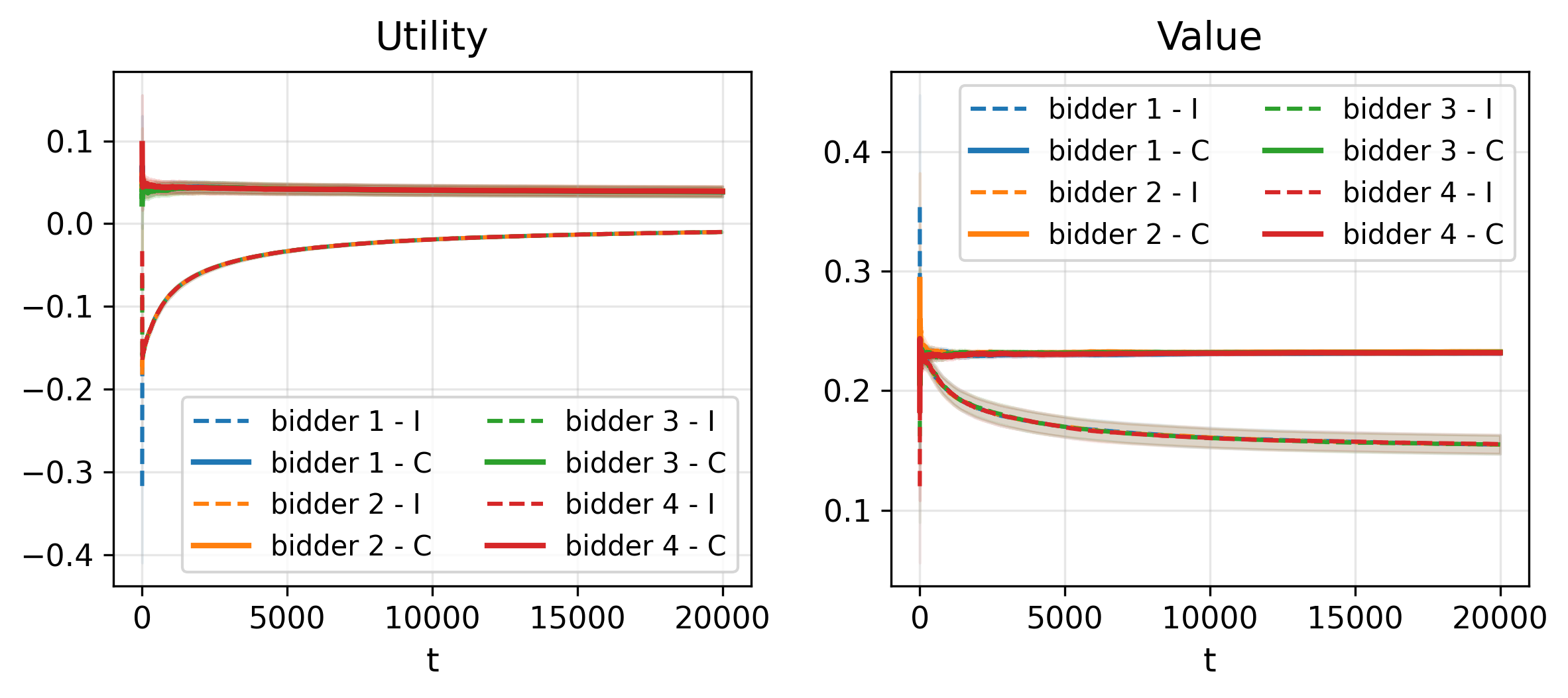}
  \caption{Experiments in real-world datasets. ($N=4$)}
  \label{fig:REAL-MERGE-K4}
\end{figure}

\section{Discussion}

In this paper, we studied coordinated auto-bidding mechanisms for value maximization under RoS constraints, to understand how cooperation among learning agents enhances efficiency in repeated auctions.
We proposed a simple yet effective coordination protocol, theoretically identified the exact condition under which coordination improves the RoS compliance for every bidder, 
and proved that coordination unconditionally increases the total value attained by the coalition compared to independent bidding.  
% yields provable guarantees for  RoS constraint compliance.
Our empirical results further validate these findings, showing higher achieved value and lower RoS violation. 
% , and consistent convergence patterns.
The demonstrated benefits of coordination over independent bidding might provide useful guidance for industrial practice.

Building on these findings, we outline several directions for future work: 
\begin{itemize}[leftmargin=1.5em]
    \item {\em Auto-bidding outside bidders:} While we assumed i.i.d.~outside bids (a reasonable assumption when outside bidders are uniformly sampled from a population every round), a natural extension % of our framework
    is to incorporate \emph{auto-bidding outside bidders}, allowing both coalition members and external participants to adjust bids dynamically.

% \paragraph{Asymmetric settings.}
% Our theoretical results can be generalized to asymmetric environments. Specifically, if for each bidder $i$ we have
% \[
% \Delta_i := \E\!\left[(v_i - d^O)\,\1\{v_i = \max_{j \in [N]} v_j\}\right] \ge C_i > 0,
% \]
% then the conclusion that coordination improves the coalition’s total achieved value remains valid. 
% Empirically, as shown in Table~\ref{tab:delta-noniid} and Figure~\ref{fig:non-iid}, bidders with larger $\Delta_i$ tend to exhibit faster convergence of their multipliers toward zero, reflecting stronger consistency with theoretical predictions.
% \gyr{Does this still account for a future direction?}
% \tao{Maybe add a Section with title ``Extension to Asymmetric Setting'' and put these discussions, together with formal mathematical results, there. Does the violation improvement result also generalize to the asymmetric setting? }

% \begin{table*}
% \begin{tabular}{ll}
% \hline
% Figure & $\Delta$ estimates \\ 
% \hline
% \ref{fig:NI1} & $\Delta_1=0.0022$, $\Delta_2=0.2183$ \\
% \ref{fig:NI2} & $\Delta_1=0.0061$, $\Delta_2=0.2018$, $\Delta_3=0.0515$ \\
% \ref{fig:NI3} & $\Delta_1=0.0269$, $\Delta_2=0.0724$, $\Delta_3=0.3441$, $\Delta_4=0.0049$, $\Delta_5=0.1706$ \\
% \hline
% \end{tabular}
% \caption{Per-bidder $\Delta_i$ estimates for non-i.i.d. experiments.}
% \label{tab:delta-noniid}
% \end{table*}

\item {\em Alternative coordination mechanisms:}
% Another direction is to explore alternative coordination mechanisms, such as highest-bid coordination, to examine how they would affect the bidders' overall welfare. Highest-bid coordination is analogous to highest-value coordination in the symmetric setting, but not in the asymmetric setting. 
While we proved that our highest-value-compete coodination mechanism is asymptotically optimal under Assumption \ref{ass:highest-value-assumption}, it is still interesting to examine whether alternative mechanisms are better when the assumption does not hold. Moreover, if bidders' private values are not accessible, then one may consider other coordination mechanisms, such as highest-bid-compete or mechanisms involving private information elicitation. 

\item {\em Other auction formats:}
Finally, it would be valuable to extend our analysis to other auction formats, especially non-truthful auctions such as first-price and generalized second-price auctions. Such extensions would help test the robustness of our theoretical insights and clarify how coordination interacts with strategic bidding incentives in broader market environments.
\end{itemize}

\bibliographystyle{ACM-Reference-Format}
\bibliography{ref}

\appendix

\section{Omitted Proofs in Section \ref{sec:RoS}}

\subsection{Proof of Observation \ref{obs:large-N-assumption-holds}}
Fix $\varepsilon\in(0,B)$. The probability $q_\varepsilon := \Pr_{v\sim F}[v > B - \eps] = 1-F(B-\varepsilon)>0$ by full support/continuity of $F$.
Let $M_N := \max\{v_1,\dots,v_N\}=v^{(N)}$ and $S_N:=v^{(N-1)}$.
The events $\{v_n > B-\varepsilon\}$ are i.i.d.\ with probability $q_\varepsilon$ and
$\sum_{n=1}^\infty \Pr(v_n>B-\varepsilon)=\infty$, so by the second Borel--Cantelli lemma,
$v_n > B-\varepsilon$ occurs infinitely often for all sufficiently large $n$ with probability one.
By monotonicity of the running maxima and of the second maximum, this implies
$M_N\ge B-\varepsilon$ and $S_N\ge B-\varepsilon$ for all sufficiently large $N$ with probability one;
hence $M_N\to B$ and $S_N\to B$ almost surely.

Define $Z_N := (S_N-d^O)_+ - (d^O-M_N)_+ \in[-B,B]$ and $\Delta_N = \E[Z_N]$.
Since $M_N,S_N\to B$ almost surely and $d^O\in[0,B]$, we have $Z_N \to (B-d^O)_+$ almost surely.
By the dominated convergence theorem,
\[
\lim_{N\to\infty} \Delta_N ~ = ~ \lim_{N\to\infty}\mathbb{E}[Z_N] ~ = ~ \mathbb{E}[\lim_{N\to\infty}Z_N] ~ = ~ \mathbb{E}\!\left[(B-d^O)_+\right] ~ \ge ~ 0,
\]
with strict inequality whenever $\Pr(d^O<B)>0$.
Therefore there exists $N_0$ such that $\Delta_N\ge 0$ for all $N\ge N_0$.

\subsection{Proof of the Second (Negative) Result in Theorem~\ref{thm:utility/RoS-constraint}}
\label{app:negative}
Suppose Assumption \ref{ass:highest-value-assumption} does not hold. 
Let the upper bound on value and outside bid be $B = 1$. 
For every bidder $i \in [N]$, we construct the following (weakly) overbidding algorithm $A$ that chooses the current bid based on the bidder's past bids:
\[
A(H_{i,t}) \;=\;
\begin{cases}
1, & \text{if there exists some } s<t \text{ with } b_{i,s}=0,\\
v_{i,t}, & \text{otherwise.}
\end{cases}
\]
% Let every other bidder $j \in [N]\setminus\{i\}$ bid truthfully: $b_{j, t} = v_{j, t}$. 
\begin{claim}
%For every bidder $i$ and
For every horizon $T\ge 1$, every bidder $i$'s expected total utility difference is 
\begin{equation}
\label{eq:exact-gap}
\mathbb{E}\!\left[U_{i, T}^{C,A}-U_{i, T}^{I,A}\right]
\;=\;
\frac{T}{N}\,\Delta
\;+\;
\frac{1- N^{-T}}{N-1}\,L.
\end{equation}
where $L \;:=\; \mathbb{E}\big[(d^O-v^{(N)})_+\big].$

In particular, if $\Delta<0$ (i.e., Assumption~\ref{ass:highest-value-assumption} fails with strict negativity), then for all
\[
T \;>\; \frac{N}{N-1}\cdot \frac{L}{-\Delta},
\]
we have $\mathbb{E}[U_{i, T}^{C,A}]<\mathbb{E}[U_{i, T}^{I,A}]$.
\end{claim}
\begin{proof}
We first note that, under independent bidding, the first case of algorithm $A$ never triggers. Indeed, since
$v_{i,t}>0$ almost surely (by continuity of $F$) and $A$ outputs either $v_{i,t}$ or $1$, bidder $i$
never submits a zero bid. Hence the condition
$\exists s<t$ with $b_{i,s}=0$ is never satisfied, so $b_{i,t}=v_{i,t}$ for every $t\in[T]$ and every bidder $i \in [N]$, hence
\[
U_{i, T}^{I,A} = U_{i, T}^{\mathrm{Truth}}, 
\]
where $U_{i, T}^{\mathrm{Truth}}$ is bidder $i$'s total utility when all bidders in $[N]$ bid truthfully. 

We now analyze the coordinated scenario. In each round $t$, only the highest-value bidder
$i_t^*$ follows $A$, while all other bidders bid $0$. Fix a bidder $i$ and define
\[
\tau_i := \min\{t\ge 1 : i\neq i_t^*\},
\]
as the first time when bidder $i$ does not have the highest value, 
so that bidder $i$ bids $0$ for the first time in round $\tau_i$, and the trigger is active from
round $\tau_i+1$ onward. By symmetry of values and independence across rounds,
\[
\Pr[\tau_i\ge t]
=
\Pr[i_1^*=i,\ldots,i_{t-1}^*=i]
=
\Big(\frac{1}{N}\Big)^{t-1}.
\]

Let $\Delta_{i,t}:=u_{i,t}^{C,A}-u_{i,t}^{\mathrm{Truth}}$. Since
$U_{i, T}^{C,A}-U_{i, T}^{I,A}=\sum_{t=1}^T \Delta_{i,t}$, it suffices to compute $\mathbb{E}[\Delta_{i,t}]$.

- If $\tau_i<t$, bidder $i$ has already bid $0$ in the past, so when selected (i.e., having the highest value $v_t^{(N)}$) in round $t$ it bids $1$.
A direct comparison with truthful bidding shows that on this event
\begin{align*}
\Delta_{i,t}
& ~ = \underbrace{\big( v_t^{(N)} - d_t^O \big) }_{\text{$u_{i, t}^{C, A} $ when $i$ is selected}} - ~ \underbrace{\big( v_t^{(N)} - \max\big\{ v_t^{(N-1)}, d_t^O \big\} \big) \cdot \1\big[v_t^{(N)} > \max\big\{ v_t^{(N-1)}, d_t^O \big\}\big]}_{\text{$u_{i, t}^{\truth}$ when $i$ has the highest value}} \\
& ~ = ~ \big( v_t^{(N)} - d_t^O \big) ~ - ~ \big( v_t^{(N)} - \max\big\{ v_t^{(N-1)}, d_t^O \big\} \big) \cdot \1\big[v_t^{(N)} >  d_t^O \big] \\
& ~ = ~ \1\big[v_t^{(N)} >  d_t^O \big] \cdot \big( \max\big\{ v_t^{(N-1)}, d_t^O \big\} - d_t^O \big) ~ + ~ \1\big[v_t^{(N)} \le  d_t^O \big] \cdot \big( v_t^{(N)} - d_t^O \big)  \\
& ~ = ~ (v_t^{(N-1)}-d_t^O)_+-(d_t^O-v_t^{(N)})_+.
\end{align*}
By symmetry, bidder $i$ is selected with probability $1/N$, so
\[
\mathbb{E}[\Delta_{i,t}\mid \tau_i<t] ~ = ~ \frac{1}{N}\Delta.
\]

- If instead $\tau_i\ge t$, bidder $i$ has not yet bid $0$, and therefore bids truthfully when selected.
In this case the only difference relative to truthful bidding is the reduced payment under
coordination, yielding
$
\Delta_{i,t}=(v_t^{(N-1)}-d_t^O)_+$,
\[
\mathbb{E}[\Delta_{i,t}\mid \tau_i\ge t]
~ = ~ 
\frac{1}{N}\mathbb{E}[(v^{(N-1)}-d^O)_+]
~ = ~
\frac{\Delta+L}{N}.
\]

Combining the two cases,
\[
\mathbb{E}[\Delta_{i,t}]
~ = ~
\Pr[\tau_i<t]\frac{\Delta}{N} + \Pr[\tau_i\ge t]\frac{\Delta+L}{N}
~ = ~
\frac{\Delta}{N}+\frac{L}{N^t}.
\]
Summing over $t=1,\ldots,T$ gives
\[
\mathbb{E}[U_{i, T}^{C,A} - U_{i,T}^{I,A}]
~ = ~
\sum_{t=1}^T \mathbb{E}[\Delta_{i,t}]
~ = ~
\frac{T}{N}\Delta+\frac{1-N^{-T}}{N-1}L,
\]
which proves \Cref{eq:exact-gap}. If $\Delta<0$, the final inequality follows immediately for
$T$ exceeding the stated threshold.
\end{proof}

\subsection{Proof of Lemma~\ref{lem:U-C-A-vs-U-truthful}}
\label{app:proof-of-lemU-C-A-vs-U-truthful}

For any sequence of values and outside bids $(\bm v_t, d_t^O)_{t=1}^T$, consider the difference between bidder $i$'s coordinated auto-bidding utility and truthful bidding utility at round $t$: 
\begin{align*}
    \Delta_{i, t} & = u_{i,t}^{C,A} -u_{i,t}^{\truth}  \\
    & = (v_{i,t}-d_t^O)\cdot\1[b_{i,t}^{C, A} \ge d_t^O]\cdot\1[v_{i,t}=\max_{j\in[N]} v_{j,t}] - (v_{i,t}-d_{i, t}^\truth)\cdot\1[v_{i,t}\ge d_{i, t}^\truth]  \nonumber 
\end{align*}
where $d_{i, t}^\truth = \max\{ d_t^O,  \max_{j\ne i} v_{j, t}\}$ is the competing bid for bidder $i$ under truthful bidding. 
We do a case analysis: 
\begin{itemize}[leftmargin=1.5em]
    \item If $v_{i,t}\neq\max_{j\in[N]}v_{j,t}$, then bidder $i$ loses under both coordinated bidding and truthful bidding, so $u_{i,t}^{C,A} = u_{i,t}^\truth = 0$ and $\Delta_{i, t} = 0$. 
    \item If $v_{i,t}=\max_{j\in [N]}v_{j,t}$, then
    \begin{align*}
        & u_{i,t}^{C,A} = (v_{i,t}-d_t^O)\cdot\1[b_{i,t}^{C, A}\ge d_t^O] \\
        & u_{i,t}^{\truth} = (v_{i,t}-d_t^\truth)\cdot\1[v_{i,t}\ge d_t^O]. 
    \end{align*}
    Because the auto-bidder is overbidding ($v_{i, t} \le b_{i, t}^{C, A}$), there are three further cases: 
    \begin{itemize}
        \item $d_t^O < v_{i, t} \le b_{i,t}^{C, A}$: in this case, $\Delta_{i, t} = d_t^\truth-d_t^O\ge0$.
        \item $v_{i, t} \le b_{i,t}^{C, A} < d_t^O$: in this case, the bidder loses under both scenarios, so $\Delta_{i, t} = 0$.
        \item $v_{i,t} \le d_t^O\le b_{i,t}^{C, A}$: in this case, the bidder loses under truthful bidding but wins under coordinated bidding, so $\Delta_{i, t} = u_{i, t}^{C, A} = v_{i,t} - d_t^O \le 0$.
    \end{itemize}
\end{itemize}

Then, we condition on any history $H_{t-1}$, randomize over $\bm v_t$ and $d_t^O$, and analyze the expectation of $\Delta_{i, t}$.  Let $F_{(N)}$ be the distribution of the largest value $v_{(N)}$ among $N$ samples from $F$. 
\begin{align*}
& \E_{\bm v_t\sim F,d_t^O\sim D}\big[ \Delta_{i, t} \mid H_{t-1} \big]\\
& = 0 +
\Pr[v_{i, t} = \max_{j\in[N]} v_{j,t}] \cdot \E_{v_{i,t}\sim F_{(N)},d_t^O\sim D}\big[ \Delta_{i, t} \mid H_{t-1} \big] \\
& = \frac{1}{N} \bigg(\Pr_{v_{i,t}\sim F_{(N)},d_t^O\sim D}\Big[v_{i,t} > d_t^O\Big] \cdot \E\Big[d_t^\truth - d_t^O \,\big|\, v_{i,t}\ge d_t^O \Big]\\
& \qquad + \Pr_{v_{i,t}\sim F_{(N)},d_t^O\sim D}\Big[v_{i,t} \le d_t^O \le b_{i,t}^{C, A}\Big] \cdot \E\Big[v_{i,t}-d_t^O \,\big|\, v_{i,t}\le d_t^O \le b_{i,t}^{C, A}\Big] \bigg) \\
& = \frac{1}{N} \bigg( \E\Big[ \mathbb{I}\big[v_{(N)} > d_t^O \big] \cdot \big(\max\{v_{(N-1)}, d_t^O \} - d_t^O \big) \Big] + \E\Big[ \mathbb{I}\big[ v_{(N)}\le d_t^O \le b_{i,t}^{C, A}\big] \cdot\big( v_{(N)}-d_t^O \big) \Big] \bigg)
\end{align*}
Because $\max\{v_{(N-1)}, d_t^O \} - d_t^O > 0$ implies $\mathbb{I}\big[v_{(N)} > d_t^O \big] = 1$, the above is equal to 
\begin{align*}
& = \frac{1}{N} \bigg( \E\Big[ \max\{v_{(N-1)}, d_t^O \} - d_t^O \Big] + \E\Big[ \mathbb{I}\big[ v_{(N)} \le d_t^O\le b_{i,t}^{C, A}\big] \cdot \big( v_{(N)}-d_t^O \big) \Big] \bigg) \\
& \ge \frac{1}{N} \bigg( \E\Big[ \max\{v_{(N-1)}, d_t^O \} - d_t^O \Big] + \E\Big[ \mathbb{I}\big[ v_{(N)} \le d_t^O\big] \cdot\big( v_{(N)}-d_t^O \big) \Big] \bigg) \\
& = \frac{1}{N} \bigg( \E\Big[ \big(v_{(N-1)} - d_t^O\big)_+ \Big] - \E\Big[ \big( d_t^O - v_{(N)} \big)_+ \Big] \bigg). 
\end{align*}
Using Assumption~\ref{ass:highest-value-assumption}, we obtain
\begin{align*}
    \E_{\bm v_t\sim F,d_t^O\sim D}\InBrackets{\Delta_{i, t} \mid H_{t-1} } ~ \ge ~ \frac{\Delta}{N}. 
\end{align*}
Therefore,
\[
\E\Big[U_{i, T}^{C, A} - U_{i, T}^{\truth}\Big] ~ = ~ \E\left[\sum_{t=1}^T \Delta_{i,t}\right] ~ \ge ~ T\cdot \frac{\Delta}{N},
\]
which proves the lemma. 

\section{Omitted Proofs in Section \ref{sec:value}}
% For simplicity, 
Denote the time-average total values among the $N$ coalition bidders by
$V_T^{C,A}:=\frac 1 T \sum_{t=1}^T V_t^{C,A}$ and
$V_T^{I,A}:=\frac 1T \sum_{t=1}^T V_{t}^{I,A}$ for the coordinated and
independent cases, respectively.

\subsection{Proof of Theorem~\ref{thm:value-wo-ass}}
\label{proof:value-wo-ass}

Let $v_{(N)} = \max\{v_1, \ldots, v_N\}$ denote the highest value among $N$ i.i.d.~samples from $F$.  Let $F_{(N)}$ be the distribution of $v_{(N)}$. 
Let $G_{(N)}(\lambda)$ be the expected utility of a single bidder with value $v_{(N)} \sim F_{(N)}$ and bidding using multiplier $\lambda > 0$, competing against the outside bid $d^O \sim D$:  
\begin{equation*}
    G_{(N)}(\lambda) := \E_{v_{(N)}\sim F_{(N)}, d^O\sim D}\Big[(v_{(N)} - d^O) \cdot \1\big[ (1+\tfrac{1}{\lambda}) v_{(N)} > d^O\big]\Big]. 
\end{equation*}

Define the corresponding single-bidder expected value as
\[
V_{(N)}(\lambda):= \mathbb E_{v_{(N)}\sim F_{(N)},\,d^O\sim D} \Big[ v_{(N)} \1 \big[ (1 + \tfrac1\lambda )v_{(N)}>d^O \big] \Big].
\]

Since \(V_{(N)}(\lambda)\) is monotone and bounded, the right limit
\(V_{(N)}(0^+):=\lim_{\lambda\downarrow0}V_{(N)}(\lambda)\) exists. Define the
stable reference multiplier by
\[
\lambda_\star := \inf \big\{\lambda>0:G_{(N)} (\lambda) \ge 0 \big\}\in[0,\infty), 
\]
and denote \(V_{(N)}(0) = V_{(N)}(0^+)\) when \(\lambda_\star=0\).

We present the monotonicity property of $G_{(N)}(\lambda)$. 

\begin{lemma}[Monotonicity of \(G_{(N)}\)]
\label{lem:G-monotone}
\label{lemma:mono}
Assume that the outside-bid distribution \(D\) admits a density \(f_D\) that is
strictly positive on \([0,B]\). Then \(G_{(N)}(\lambda)\) is continuous and strictly increasing on \((0,\infty)\). Moreover, \(\{\lambda > 0: G_{(N)}(\lambda)\ge0\}\) is nonempty. Hence, if \(\lambda_\star>0\), then \(G_{(N)}(\lambda_\star)=0\), while if
\(\lambda_\star=0\), then \(G_{(N)}(\lambda)>0\) for every \(\lambda>0\).
\end{lemma}

\begin{proof}
Continuity follows from dominated convergence and the continuity of \(D\).
Write the single-bidder's expected utility as
\(G_{(N)}(\lambda)=\mathbb E_{v\sim F_{(N)}}G(\lambda;v)\), where \(G(\lambda;v)\) is
the bidder's expected utility conditioning on having value \(v\). Namely,
\begin{equation*}
G(\lambda;v) = \mathbb E_{d^O\sim D}
\left[ (v-d^O)\,\1\left\{\left(1+\tfrac1\lambda\right)v>d^O\right\} \right] = v\,x\left(\left(1+\tfrac1\lambda\right)v\right) -p\left(\left(1+\tfrac1\lambda\right)v\right),
\end{equation*}
where \(x(b)\) is the probability of winning when the bidder bids \(b\), and
\(p(b)\) is the expected payment. Extending \(f_D\) by zero outside
\([0,B]\), we have
\begin{align*}
    x(b)=\int_0^b f_D(z)\,\mathrm{d}z,
    \qquad
    p(b)=\int_0^b z f_D(z)\,\mathrm{d}z .
\end{align*}
Since the second-price auction is truthful, Myerson's payment identity gives
\begin{equation*}
    p(b)=x(b)b-\int_0^b x(z)\,\mathrm{d}z .
\end{equation*}
Therefore,
\begin{align*}
G(\lambda;v)
& = v x\left(\left(1+\tfrac1\lambda\right)v\right) - \left(1+\tfrac1\lambda\right)v x\left(\left(1+\tfrac1\lambda\right)v\right) + \int_0^{(1+1/\lambda)v}x(z)\,\mathrm{d}z \\
&= -\frac1\lambda v x\left(\left(1+\tfrac1\lambda\right)v\right) + \int_0^{(1+1/\lambda)v}x(z)\,\mathrm{d}z .
\end{align*}
Taking the derivative with respect to \(\lambda\),
\begin{align*}
\frac{\partial G(\lambda;v)}{\partial\lambda}
&= \frac{1}{\lambda^2}v x\left(\left(1+\tfrac1\lambda\right)v\right) - \frac1\lambda v x'\left(\left(1+\tfrac1\lambda\right)v\right)
  \left(-\frac{v}{\lambda^2}\right)
+x\left(\left(1+\tfrac1\lambda\right)v\right)
  \left(-\frac{v}{\lambda^2}\right) \\
&= \frac{v^2}{\lambda^3}x'\left(\left(1+\tfrac1\lambda\right)v\right) \\
&= \frac{v^2}{\lambda^3}f_D\left(\left(1+\tfrac1\lambda\right)v\right).
\end{align*}
Integrating over \(v\sim F_{(N)}\),
\begin{align*}
\frac{\mathrm{d}G(\lambda)}{\mathrm{d}\lambda} = \mathbb E_{v\sim F_{(N)}} \left[ \frac{v^2}{\lambda^3} f_D\left(\left(1+\tfrac1\lambda\right)v\right)
\right].
\end{align*}
Because \(F_{(N)}\) has full support on \([0,B]\) and \(f_D(z)>0\) on
\([0,B]\), the integrand is positive with positive probability: for example,
\(v\in(0,B/(1+1/\lambda))\) has positive probability, and then
\((1+1/\lambda)v\in(0,B)\). Hence
\(\mathrm{d}G(\lambda)/\mathrm{d}\lambda>0\), so \(G\) is strictly increasing.

Finally, because $F_{(N)}$ and $D$ have ful support on $[0, B]$, 
\[
\lim_{\lambda\to\infty}G(\lambda)
=
\mathbb E\left[(v_{(N)}-d^O)\1\{v_{(N)}>d^O\}\right]>0,
\]
so \(\{\lambda>0:G(\lambda)\ge0\}\) is nonempty. The remaining claims
follow from continuity and strict monotonicity.
\end{proof}

For the MD-\(h\) update, write \(y_{i,t}:=h'(\lambda_{i,t}) \in (-\infty, +\infty) \) for the mirror
coordinate of bidder \(i\)'s multiplier. By first-order optimality, the mirror-descent update rule is 
\begin{equation}
y_{i,t+1}=y_{i,t}-\alpha g_{i,t}.
\label{eq:md-y-update}
\end{equation}

\subsection*{Step 1: Coordinated value converges to \(V_{(N)}(\lambda_\star)\)}

In this step all quantities refer to the coordinated process. Let
\(H_{t-1}\) denote the joint history before round \(t\), and define the drift in
mirror coordinates by
\[ \phi(y):=\tfrac1N G_{(N)}\left((h')^{-1}(y)\right), \quad \forall y\in\mathbb R.\]
Under coordination, bidder \(i\) is active if and only if it is the
highest-value bidder in the coalition. Since the values are i.i.d., each bidder is the highest-value bidder with probability $1/N$. So conditioning
on the history \(H_{t-1}\),
\begin{equation}
\mathbb E\big[ g_{i,t}\mid H_{t-1}] = \tfrac1N G_{(N)}(\lambda_{i,t}) = \phi(y_{i,t}).
\label{eq:coord-drift}
\end{equation}
Also, in the coordinated process, \(g_{i,t}\in[-B,B]\) almost surely.

\begin{lemma}[Interior occupation bound]
\label{lem:coord-interior}
Assume \(\lambda_\star>0\). Let \(y_\star:=h'(\lambda_\star)\). For every
bidder \(i\in[N]\) and every \(\delta>0\),
\[
\frac1T\sum_{t=1}^T \mathbb P \big(|y_{i,t}-y_\star|\ge\delta \big) \le \frac{(y_{i,1}-y_\star)^2+B^2}{2\delta c_\delta}\cdot\frac1{\sqrt T},
\]
where
\[
c_\delta := \min\{-\phi(y_\star - \delta),\,\phi(y_\star+\delta)\}>0.
\]
\end{lemma}

\begin{proof}
Since \(\phi\) is strictly increasing and \(\phi(y_\star)=0\), \(c_\delta>0\).
For all \(y\) with \(|y-y_\star|\ge\delta\),
\begin{equation}
(y-y_\star)\phi(y)\ge \delta c_\delta.
\label{eq:drift-away}
\end{equation}
Let \(W_t:=(y_{i,t}-y_\star)^2\). By Equation (\ref{eq:md-y-update}),
\[
W_{t+1}
=
W_t-2\alpha(y_{i,t}-y_\star)g_{i,t}
+\alpha^2g_{i,t}^2.
\]
Taking conditional expectation and using Equation (\ref{eq:coord-drift}) and
\(g_{i,t}^2\le B^2\),
\[
\mathbb E[W_{t+1}\mid H_{t-1}]
\le
W_t-2\alpha(y_{i,t}-y_\star)\phi(y_{i,t})+\alpha^2B^2.
\]
Summing over \(t=1,\ldots,T\), using \(W_{T+1}\ge0\), and using
\(\alpha^2T=1\), we obtain
\[
2\alpha \sum_{t=1}^T \mathbb E\big[(y_{i,t}-y_\star)\phi(y_{i,t})\big] \le W_1+B^2.
\]
By Equation (\ref{eq:drift-away}),
\[
\mathbb E\big[(y_{i,t}-y_\star)\phi(y_{i,t})\big] \ge \delta c_\delta\, \mathbb P\big(|y_{i,t}-y_\star|\ge\delta\big).
\]
Dividing by \(2\alpha T\) gives the claim.
\end{proof}

\begin{lemma}[Boundary occupation bound]
\label{lem:coord-boundary}
Assume \(\lambda_\star=0\). For every bidder \(i\in[N]\) and any
\(\delta>0\),
\[
\frac1T\sum_{t=1}^T \mathbb P\big(\lambda_{i,t}\ge\delta \big) = O(T^{-1/2}).
\]
\end{lemma}

\begin{proof}
Since \(\lambda_\star=0\), Lemma~\ref{lem:G-monotone} implies \(G_{(N)}(\lambda)>0\) for all \(\lambda>0\). Hence \(\phi(y)>0\) for any \(y\in\mathbb R\).

Fix \(M\in\mathbb R\). It suffices to prove
\[
\frac1T\sum_{t=1}^T \mathbb P \big(y_{i,t}\ge M \big) = O(T^{-1/2}).
\]
Let \(W_t:=\exp(y_{i,t})\). By Equation (\ref{eq:md-y-update}), \(W_{t+1}=W_t\exp(-\alpha g_{i,t})\). Using \(e^{-x}\le 1-x+x^2e^{|x|}/2\) and \(|g_{i,t}|\le B\), we have
$$
\exp(-\alpha g_{i,t})-1\le -\alpha g_{i,t} + \frac{1}{2}\alpha^2 g_{i,t}^2 e^{\alpha|g_{i,t}|}\le -\alpha g_{i,t} + \frac{B^2e^B}{2}\alpha^2
$$
Therefore,
$$
\mathbb E\big[W_{t+1}-W_t\mid H_{t-1} \big] = W_t \mathbb E\big[\exp(- \alpha g_{i, t}) - 1 \mid H_{t-1} \big] \le - \alpha \phi(y_{i,t}) W_t + \frac{B^2e^B}{2}\alpha^2W_t.
$$
On \(\{y_{i,t}\ge M\}\), monotonicity gives \(\phi(y_{i,t})\ge\phi(M)>0\), while outside this event \(\phi(y_{i,t})>0\). Thus
\begin{equation}
\mathbb E\big[W_{t+1}-W_t\big] \le - \alpha\phi(M)\mathbb E\big[W_t\1\{y_{i,t}\ge M\}\big] + \frac{B^2e^B}{2}\alpha^2\mathbb E\big[W_t\big] .
\label{eq:exp-potential}
\end{equation}
Hoeffding's lemma and Equation (\ref{eq:coord-drift}) give
$$
\mathbb E\big[\exp(-\alpha g_{i,t})\mid H_{t-1}\big] \le \exp\left( -\alpha\phi(y_{i,t})+\frac{\alpha^2B^2}{2} \right) \le \exp\left(\frac{\alpha^2B^2}{2}\right).
$$

Therefore, since $\alpha^2 T = 1$,
$$
\sup_{t\le T+1} \mathbb E[W_t] \le \mathbb E[W_1] \exp\left(\frac{B^2}{2}\right). 
$$

% In particular, if $y_{i,1}=0$, then
% $$
% \sup_{t\le T+1} \mathbb E[W_t] \le \exp\left(\frac{B^2}{2}\right)
% $$

Summing (\ref{eq:exp-potential}), using
\(W_{T+1}\ge0\), and using \(\alpha^2T=1\), we get
$$
\alpha\phi(M) \sum \nolimits_{t=1}^T \mathbb E\big[W_t\1\{y_{i,t}\ge M\}\big] \le \mathbb E[W_1]+\frac{B^2e^B}{2} \mathbb E[W_1] \exp\left(\frac{B^2}{2}\right).
$$
Since \(W_t\ge e^{M}\) on \(\{y_{i,t}\ge M\}\),
$$
\sum \nolimits_{t=1}^T\mathbb P\big(y_{i,t}\ge M\big)\le\frac{e^{-M}}{\alpha \phi(M)}\left(\mathbb E[W_1]+\frac{B^2e^B}{2} E[W_1] \exp\left(\frac{B^2}{2}\right)\right)=O(\alpha^{-1})=O(\sqrt T) . 
$$
Dividing by \(T\) gives the desired occupation bound. Finally, taking
\(M=h'(\delta)\) proves the lemma.
\end{proof}

\begin{lemma}[Coordinated value limit]
\label{lem:coord-value-limit}
Under coordinated bidding,
\[
\lim_{T\to\infty}\mathbb E\big[ V_T^{C,A} \big] = V_{(N)}(\lambda_\star).
\]
\end{lemma}

\begin{proof}
Let \(i_t^\star:=\operatorname*{arg\,max}_{i\in[N]}v_{i,t}\) be the highest-value bidder in round $t$. Under coordination,
\[
V_t^{C,A} = v_{(N),t} \cdot \1 \left\{(1+\tfrac1{\lambda_{i_t^\star,t}})v_{(N),t}>d_t^O\right\}.
\]
Hence the expected $T$-round total value is 
\begin{equation}
\mathbb E\big[ V_T^{C,A} \big] = \mathbb E\left[
\frac1T\sum\nolimits_{t=1}^T
V_{(N)}(\lambda_{i_t^\star,t})
\right]. 
\label{eq:coord-value-as-average}
\end{equation}

- First suppose \(\lambda_\star>0\). Fix \(\varepsilon>0\). By continuity of
\(V_{(N)}\) at \(\lambda_\star\), choose \(\delta>0\) such that
\[
\left|V_{(N)}((h')^{-1}(y))-V_{(N)}(\lambda_\star)\right|
\le \varepsilon
\quad\text{whenever } |y-y_\star|<\delta.
\]
Using \(0\le V_{(N)}(\cdot)\le B\),
\[
\left|\mathbb E\big[ V_T^{C,A} \big] - V_{(N)}(\lambda_\star)\right| \le \varepsilon + 2B\,
\mathbb E\left[ \frac1T\sum \nolimits_{t=1}^T \1\{|y_{i_t^\star,t}-y_\star|\ge\delta\} \right].
\]
Conditional on \(H_{t-1}\), the index \(i_t^\star\) is uniform on \([N]\).
Thus
\[
\mathbb E\big[
\1\{|y_{i_t^\star,t}-y_\star|\ge\delta\}
\mid H_{t-1}\big] = \frac1N\sum \nolimits_{i=1}^N
\1\{|y_{i,t}-y_\star|\ge\delta\}.
\]
Lemma~\ref{lem:coord-interior} implies that the second term vanishes as
\(T\to\infty\). Letting \(\varepsilon\downarrow0\) proves the interior case.

- Then suppose \(\lambda_\star=0\). Fix \(\varepsilon>0\). By the definition of
\(V_{(N)}(0^+)\), choose \(\delta>0\) such that
\[
|V_{(N)}(\lambda)-V_{(N)}(0^+)| \le\varepsilon \qquad \forall \lambda\in(0,\delta].
\]
Then, by Equation (\ref{eq:coord-value-as-average}),
\[
\left|\mathbb E\big[V_T^{C,A}\big] - V_{(N)}(0^+)\right| \le \varepsilon + 2B\,
\mathbb E\left[
\frac1T\sum \nolimits_{t=1}^T
\1\{\lambda_{i_t^\star,t}>\delta\}
\right].
\]
Again, conditional on \(H_{t-1}\),
\[
\mathbb E\big[
\1\{\lambda_{i_t^\star,t}>\delta\}
\mid H_{t-1}\big] = \frac1N\sum \nolimits_{i=1}^N
\1\{\lambda_{i,t}>\delta\}.
\]
Lemma~\ref{lem:coord-boundary} implies that the second term vanishes as $T\to\infty$. Since
\(\varepsilon\) is arbitrary, the boundary case follows.
\end{proof}

\subsection*{Step 2: Virtual RoS feasibility under independent bidding}

In this step all quantities refer to the independent bidding case, and we sometimes omit the superscript \(I,A\). Define
\[
\bar b_t:=\max_{i\in[N]} b_{i,t},
\qquad
\bar V_t:=v_{(N),t}\1\{\bar b_t>d_t^O\},
\qquad
\bar g_t:=(v_{(N),t}-d_t^O)\1\{\bar b_t>d_t^O\}.
\]
We view \(\bar b_t\) as the bid of a virtual bidder who has value
\(v_{(N),t}\), competes against the outside bid $d_t^O$ only, and obtains value $\bar V_t$ and utility $\bar g_t$. Ties have probability zero under the continuity assumptions, and
all pointwise statements below are understood almost surely.

For each bidder \(i \in [N]\), define the clipped realized utility
\[
\widetilde g_{i,t}:=\max\{g_{i,t},-B\}.
\]

\begin{lemma}[Virtual utility dominates clipped realized utilities]
\label{lem:virtual-clipped-domination}
For every round \(t\), almost surely,
\[
V_t^{I,A}\le \bar V_t, \qquad \bar g_t\ge \sum \nolimits_{i=1}^N \widetilde g_{i,t}.
\]
\end{lemma}

\begin{proof}
Fix a round \(t\). If \(\bar b_t\le d_t^O\), then no coalition bidder wins
against the outside bid. Hence \(V_t^{I,A}=0\), \(g_{i,t}=0\) for all \(i\),
and \(\bar V_t=\bar g_t=0\).

Now suppose \(\bar b_t>d_t^O\). Let \(j_t\) be the coalition winner, and let
\(b^{(2)}_t\) be the second-highest coalition bid. Then
\[
V_t^{I,A}=v_{j_t,t}\le v_{(N),t}=\bar V_t.
\]
Only bidder \(j_t\) has nonzero realized utility, and
\[
g_{j_t,t} = v_{j_t,t}-\max\{d_t^O,b^{(2)}_t\}.
\]
Since \(v_{j_t,t}\le v_{(N),t}\) and
\(\max\{d_t^O,b^{(2)}_t\}\ge d_t^O\),
\[
g_{j_t,t}\le v_{(N),t}-d_t^O=\bar g_t.
\]
Also \(\bar g_t\ge -B\), because \(v_{(N),t},d_t^O\in[0,B]\). Therefore
\[
\sum \nolimits_{i=1}^N\widetilde g_{i,t} = \max\{g_{j_t,t},-B\} \le \bar g_t.
\]
\end{proof}

\begin{lemma}[Self-bounding of clipped utilities]
\label{lem:clipped-self-bound}
For every bidder \(i\) and every round \(t\),
\[
-B\le \widetilde g_{i,t}\le B, \qquad \widetilde g_{i,t}\ge -\tfrac{B}{\lambda_{i,t}},
\]
and in mirror coordinates,
\[
y_{i,t+1}\ge y_{i,t}-\alpha \widetilde g_{i,t}.
\]
\end{lemma}

\begin{proof}
The lower bound \(\widetilde g_{i,t}\ge -B\) holds by definition. Also
\(g_{i,t}\le v_{i,t}x_{i,t}\le B\), so \(\widetilde g_{i,t}\le B\).

If bidder \(i\) loses, then \(g_{i,t}=0\), and hence
\(\widetilde g_{i,t}=0\ge -B/\lambda_{i,t}\). If bidder \(i\) wins, then the competing bid is at most its bid:
\[
d_{i,t}\le b_{i,t} = \left(1+\tfrac1{\lambda_{i,t}}\right)v_{i,t}.
\]
Thus
\[
g_{i,t} = v_{i,t}-d_{i,t} \ge v_{i,t}- \left(1+\tfrac1{\lambda_{i,t}}\right)v_{i,t} = -\tfrac{v_{i,t}}{\lambda_{i,t}} \ge -\tfrac{B}{\lambda_{i,t}}.
\]
Since \(\widetilde g_{i,t}\ge g_{i,t}\), this proves $\widetilde g_{i, t} \ge - \tfrac{B}{\lambda_{i, t}}$.

Recall that the MD update is \(y_{i,t+1}=y_{i,t}-\alpha g_{i,t}\). Because
\(\widetilde g_{i,t}\ge g_{i,t}\),
\[
y_{i,t+1}=y_{i,t}-\alpha g_{i,t}
\ge
y_{i,t}-\alpha\widetilde g_{i,t}.
\]
\end{proof}

\begin{lemma}[Deterministic clipped violation bound]
\label{lem:clipped-violation}
Fix bidder \(i\). For any \(L>\lambda_{i,1}\),
\[
-\sum \nolimits_{t=1}^T \widetilde g_{i,t} \le \frac{h'(L)-h'(\lambda_{i,1})}{\alpha} + B + \frac{BT}{L}.
\]
\end{lemma}

\begin{proof}
Define
\[
S_t:=-\sum \nolimits_{s=1}^t \widetilde g_{i,s}, \qquad S_0:=0.
\]
By Lemma~\ref{lem:clipped-self-bound},
\[
y_{i,t+1}\ge y_{i,t}-\alpha \widetilde g_{i,t}.
\]
Iterating gives
\[
y_{i,t}\ge y_{i,1}+\alpha S_{t-1}
\qquad\forall t\ge1.
\tag{6}
\label{eq:y-lower-by-S}
\]
Let
\[
A_L:=\frac{h'(L)-h'(\lambda_{i,1})}{\alpha}.
\]
If \(S_T\le A_L\), the claim is immediate. Otherwise, define the last time at
which the cumulative clipped violation is still below \(A_L\):
\[
\tau:=\max \big\{t\in\{0,1,\ldots,T\}:S_t\le A_L \big\}.
\]
Then \(\tau\le T-1\), \(S_\tau\le A_L\), and 
\[
S_{t-1}>A_L\qquad\forall t\ge \tau+2.
\]
For every \(t\ge\tau+2\), Inequality (\ref{eq:y-lower-by-S}) gives
\[
y_{i,t} > y_{i,1}+\alpha A_L = h'(\lambda_{i,1})+h'(L)-h'(\lambda_{i,1}) = h'(L).
\]
Since \(h'\) is strictly increasing, \(\lambda_{i,t}>L\). Hence, by
Lemma~\ref{lem:clipped-self-bound},
\[
-\widetilde g_{i,t}\le \frac{B}{\lambda_{i,t}}<\frac{B}{L} \qquad \forall t\ge\tau+2.
\]
Using also \(-\widetilde g_{i,\tau+1}\le B\), we get
\[
S_T = S_\tau + (-\widetilde g_{i,\tau+1}) + \sum \nolimits_{t=\tau+2}^T(-\widetilde g_{i,t}) \le A_L+B+\frac{BT}{L}.
\]
Since \(S_T=-\sum_{t=1}^T\widetilde g_{i,t}\), the proof is complete.
\end{proof}

\begin{lemma}[Virtual bidder's RoS violation]
\label{lem:virtual-ros}
There exists a deterministic sequence \(R_T=o(T)\) such that, pathwise,
\[
\sum \nolimits_{t=1}^T \bar g_t\ge -N R_T.
\]
\end{lemma}

\begin{proof}
Because \(\lambda_{i,1}=1\) and \(\alpha=1/\sqrt T\), define
\[
L_T:=(h')^{-1}\left(h'(1)+T^{1/4}\right),
\qquad
R_T:= \frac{h'(L_T) - h'(1)}{\alpha} + B + \frac{BT}{L_T} = T^{3/4}+B+\frac{BT}{L_T}. 
\]
Since \(h'\) is onto and strictly increasing, \(L_T\to\infty\) as $T\to\infty$, and hence
\[
\frac{R_T}{T} = T^{-1/4}+\frac{B}{T}+\frac{B}{L_T} \to 0.
\]
Thus \(R_T=o(T)\).

Applying Lemma~\ref{lem:clipped-violation} with \(L=L_T\) gives, for every
bidder \(i\),
\[
\sum \nolimits_{t=1}^T \widetilde g_{i,t}\ge -R_T.
\]
By Lemma~\ref{lem:virtual-clipped-domination},
\[
\sum \nolimits_{t=1}^T \bar g_t \ge
\sum \nolimits_{t=1}^T\sum_{i=1}^N\widetilde g_{i,t} = \sum \nolimits_{i=1}^N\sum \nolimits_{t=1}^T\widetilde g_{i,t} \ge - NR_T.
\]
\end{proof}

\subsection*{Step 3: Dual-envelope upper bound on independent value}

For \(\lambda>0\), define
\[
\Phi_\lambda(v,b) := \mathbb E_{d^O\sim D}
\left[ v\1\{b>d^O\} + \lambda(v-d^O)\1\{b>d^O\} \right], 
\]
\[
\Psi_\lambda(v):=\sup_{b\ge0}\Phi_\lambda(v,b),
\qquad \mathcal D(\lambda):=\mathbb E_{v\sim F_{(N)}}[\Psi_\lambda(v)].
\]

\begin{lemma}[Dual envelope]
\label{lem:dual-envelope}
For every \(\lambda>0\),
\[
\mathcal D(\lambda)=V_{(N)}(\lambda)+\lambda G_{(N)}(\lambda).
\]
Moreover,
\[
\inf_{\lambda>0}\mathcal D(\lambda)=V_{(N)}(\lambda_\star).
\]
\end{lemma}

\begin{proof}
Fix \(v\in[0,B]\), \(\lambda>0\), and set
\[
c_\lambda(v):=\left(1+\tfrac1\lambda\right)v.
\]
We can rewrite
\[
\Phi_\lambda(v,b) = \mathbb E_{d^O\sim D}
\left[ \lambda(c_\lambda(v)-d^O)\1\{d^O<b\} \right].
\]
Increasing \(b\) up to \(c_\lambda(v)\) only adds nonnegative contributions,
whereas increasing \(b\) beyond \(c_\lambda(v)\) only adds nonpositive
contributions. Hence \(b=c_\lambda(v)\) maximizes \(\Phi_\lambda(v,b)\), and
\[
\Psi_\lambda(v)=\Phi_\lambda(v,c_\lambda(v)).
\]
Taking expectation over \(v\sim F_{(N)}\) yields
\[
\mathcal D(\lambda)=V_{(N)}(\lambda)+\lambda G_{(N)}(\lambda).
\]
It remains to prove the infimum identity. Consider the single-bidder primal
problem
\[
\mathrm{OPT}_{\mathrm{single}} := \sup_{\pi}
\big\{ \mathbb E[V^\pi]:
\mathbb E[g^\pi]\ge0
\big\},
\]
where the supremum is over measurable bidding rules \(\pi:[0,B]\to\mathbb R_+\), and
\[
V^\pi:=v_{(N)}\1\{\pi(v_{(N)})>d^O\},
\qquad
g^\pi:=(v_{(N)}-d^O)\1\{\pi(v_{(N)})>d^O\}.
\]
For any feasible \(\pi\) and any \(\lambda>0\),
\[
\mathbb E[V^\pi]\le\mathbb E[V^\pi+\lambda g^\pi]\le\mathcal D(\lambda).
\]
Thus
\[
\mathrm{OPT}_{\mathrm{single}}\le\inf_{\lambda>0}\mathcal D(\lambda).
\]
If \(\lambda_\star>0\), then \(G_{(N)}(\lambda_\star)=0\). The bidding rule
\[
\pi_\star(v):=\big(1+\tfrac1{\lambda_\star}\big)v
\]
is feasible and attains value \(V_{(N)}(\lambda_\star)\). Also
\[
\mathcal D(\lambda_\star) = V_{(N)}(\lambda_\star)+\lambda_\star G_{(N)}(\lambda_\star) = V_{(N)}(\lambda_\star).
\]
Therefore
\[
V_{(N)}(\lambda_\star) \le \mathrm{OPT}_{\mathrm{single}} \le \inf_{\lambda>0}\mathcal D(\lambda) \le \mathcal D(\lambda_\star) = V_{(N)}(\lambda_\star),
\]
so equality holds throughout.

If \(\lambda_\star=0\), then \(G_{(N)}(\lambda)\ge0\) for every \(\lambda>0\), so \(G_{(N)}(0^+):=\lim_{\lambda\downarrow0}G_{(N)}(\lambda)\ge0\). The always-win rule \(\pi_0(v)\equiv 2B\) is feasible and attains value \(\mathbb E[v_{(N)}]=V_{(N)}(0^+)\). Hence
\[
V_{(N)}(0^+)
\le
\mathrm{OPT}_{\mathrm{single}}
\le
\inf_{\lambda>0}\mathcal D(\lambda).
\]
On the other hand, \(|G_{(N)}(\lambda)|\le B\), and therefore
\[
\inf_{\lambda>0}\mathcal D(\lambda) \le \lim_{\lambda\downarrow0} \left(V_{(N)}(\lambda)+\lambda G_{(N)}(\lambda)\right) = V_{(N)}(0^+).
\]
Thus \(\inf_{\lambda>0}\mathcal D(\lambda)=V_{(N)}(0^+)\), which is \(V_{(N)}(\lambda_\star)\) by definition.
\end{proof}

\begin{lemma}[Independent upper bound]
\label{lem:independent-upper}
Under independent bidding,
\[
\limsup_{T\to\infty}\mathbb E\big[ V_T^{I,A} \big] \le V_{(N)}(\lambda_\star).
\]
\end{lemma}

\begin{proof}
By Lemma~\ref{lem:virtual-clipped-domination}, \(V_t^{I,A}\le\bar V_t\) almost surely, so it suffices to upper bound the time-average value of the virtual
bidder.

Fix \(\lambda>0\), and let
\(\mathcal G_t:=\sigma(H_{t-1},v_{1,t},\ldots,v_{N,t})\) be the information after coalition values are realized but before the outside bid is drawn. Conditioning on \(\mathcal G_t\), both \(v_{(N),t}\) and \(\bar b_t\) are fixed, while \(d_t^O\sim D\) is fresh and independent. Hence
\[
\mathbb E\big[\bar V_t+\lambda\bar g_t\mid\mathcal G_t\big] = \Phi_\lambda(v_{(N),t},\bar b_t) \le \Psi_\lambda(v_{(N),t}).
\]
Taking expectations and using \(v_{(N),t}\sim F_{(N)}\),
\[
\mathbb E\big[ \bar V_t+\lambda\bar g_t \big]\le \mathcal D(\lambda).
\]
Summing over \(t=1,\ldots,T\) and dividing by \(T\),
\[
\mathbb E\left[\frac1T\sum \nolimits_{t=1}^T\bar V_t\right] \le \mathcal D(\lambda) - \lambda\, \mathbb E\left[\frac1T\sum\nolimits_{t=1}^T\bar g_t\right].
\]
By Lemma~\ref{lem:virtual-ros},
\[
\sum \nolimits_{t=1}^T\bar g_t\ge -NR_T \qquad\text{pathwise},
\]
with \(R_T=o(T)\). Therefore
\[
\mathbb E\left[\frac1T\sum \nolimits_{t=1}^T\bar V_t\right] \le \mathcal D(\lambda)+\lambda\frac{NR_T}{T}.
\]
Taking \(\limsup_{T\to\infty}\) gives
\[
\limsup_{T\to\infty} \mathbb E\left[\frac1T\sum\nolimits_{t=1}^T\bar V_t\right] \le \mathcal D(\lambda)
\qquad \forall\lambda>0.
\]
Taking the infimum over \(\lambda>0\) and using
Lemma~\ref{lem:dual-envelope},
\[
\limsup_{T\to\infty} \mathbb E\left[\frac1T\sum \nolimits_{t=1}^T\bar V_t\right] \le \inf_{\lambda>0}\mathcal D(\lambda) = V_{(N)}(\lambda_\star).
\]
Since \(V_t^{I,A}\le\bar V_t\) almost surely for every \(t\), the same upper
bound holds for \(V_T^{I,A}\).
\end{proof}

\subsection*{Step 4: Conclusion}

Combining Lemma~\ref{lem:coord-value-limit} and
Lemma~\ref{lem:independent-upper}, we obtain
\[
\liminf_{T\to\infty}
\mathbb E\left[
V_T^{C,A}-V_T^{I,A}
\right]
\ge
V_{(N)}(\lambda_\star)-V_{(N)}(\lambda_\star)
=
0.
\]
This proves Theorem \ref{thm:value-wo-ass}.

\subsection{Proof of Theorem~\ref{thm:value}}

\label{app:proof-of-value}

\paragraph{\bf Reduction to a single-bidder scenario.}
To prove Theorem \ref{thm:value}, we first draw a connection between the $N$-bidder coordination scenario and the scenario where a \emph{single} bidder, whose value equals the highest value among the $N$ bidders, competes against the outside bid.

Formally, let $v_{(N)} = \max\{v_1, \ldots, v_N\}$ denote the highest value among $N$ i.i.d.~samples from $F$.  Let $F_{(N)}$ be the distribution of $v_{(N)}$. 
Let $G_{(N)}(\lambda)$ be the expected utility of a single bidder with value $v_{(N)} \sim F_{(N)}$ and bidding using multiplier $\lambda > 0$, competing against the outside bid $d^O \sim D$:  
\begin{align*}
    G_{(N)}(\lambda)=\E_{v_{(N)}\sim F_{(N)}, d^O\sim D}\Big[(v_{(N)} - d^O) \cdot \1\big[ (1+\tfrac{1}{\lambda}) v_{(N)} > d^O\big]\Big]. 
\end{align*}

Let $G_i^C(\lambda_i)$ be the expected utility of bidder $i$ with value $v_i\sim F$ and multiplier $\lambda_i > 0$ in the $N$-bidder coordinated scenario: 
\begin{align*}
    G_i^C(\lambda_i) = \E_{v_1, \ldots, v_N \sim F, d^O\sim D}\big[ (v_i - d^O) x_i(\bm v, \lambda_i) \big]
\end{align*}
where $x_i(\bm v, \lambda_i) = 1$ if $v_i = \max_{j \in [N]} \{v_j\}$ and $(1+1/\lambda_i) v_i > d^O$, and $0$ otherwise.  

We show that the $N$-bidder coordinated scenario is ``equivalent'' to the single-bidder scenario in the following sense:
\begin{observation}
\label{observation:utility-1/N}
    $G_i^C(\lambda_i) = \frac{1}{N} G_{(N)}(\lambda_i)$. 
\end{observation}
\begin{proof}
In the coordinated scenario, bidder $i$ obtains utility $0$ if its value is not the highest among the $N$ bidders. The probability that bidder $i$ has the highest value is $1/N$ because the $N$ bidders' values are i.i.d. So, 
\begin{align*}
    G_i^C(\lambda_i) & = 0 + \tfrac{1}{N}\E_{v_i \sim F| v_i=v_{(N)}}\big[ (v_i - d^O)\1[(1+1/\lambda_i) v_i > d^O] \big] \\
    & = \tfrac{1}{N}\E_{v_{(N)} \sim F_{(N)}}\big[ (v_{(N)} - d^O)\1[(1+1/\lambda_i) v_{(N)} > d^O] \big] \\
    & = \tfrac{1}{N} G_{(N)}(\lambda_i). \qedhere 
\end{align*}
\end{proof}

Since the $N$-bidder coordinated scenario is equivalent to a single-bidder scenario, we can now analyze the single-bidder scenario. 

In addition to Lemma~\ref{lemma:mono}, we present the following useful lemma regarding the expected utility function $G_{(N)}(\lambda)$ of the single bidder, under Assumption~\ref{ass:highest-value-assumption}.

\begin{lemma}
\label{lem:utility-at-least-Delta}
Under Assumption \ref{ass:highest-value-assumption}, % for any $\lambda > 0$,
$G_{(N)}(\lambda) \ge \Delta > 0$. 
\end{lemma}
\begin{proof}
By Lemma \ref{lemma:mono}, $G_{(N)}(\lambda)$ is increasing in $\lambda$. So, 
\begin{align*}
    G_{(N)}(\lambda) & \ge \lim_{\mu \to 0^+} G_{(N)}(\mu) \\
    & = \E_{v_{(N)} \sim F_{(N)}, d^O \sim D}\big[ v_{(N)} - d^O \big] \\
    & = \E\big[ (v_{(N)} - d^O)_+ \big] - \E\big[ (d^O - v_{(N)})_+ \big] \\
    & \ge \E\big[ (v_{(N - 1)} - d^O)_+ \big] - \E\big[ (d^O - v_{(N)})_+ \big] \\
    & \ge \Delta > 0
\end{align*} 
where the last ``$\ge$'' follows from Assumption \ref{ass:highest-value-assumption}. 
\end{proof}

Intuitively, Lemma \ref{lem:utility-at-least-Delta} shows that the expected utility obtained by the highest-value coalition bidder is always positive (by a margin of $\Delta > 0$). This means that the bidder's multiplier $\lambda_t$, following mirror descent updates, will decrease in expectation every round and eventually converge to $0$.  We formalize this intuition below. 

\paragraph{\bf Coordinated multipliers converge to $0$.}
Using the above characterizations, we show an interesting fact: the multipliers $\lambda_{i, t}^C$ of auto-bidders running Algorithm \ref{alg:general-mirror-descent} in the coordinated scenario \emph{converge to $0$ as $t\to\infty$}.
This means that the coalition will eventually submit arbitrarily high bids $b_{i, t}^C = (1+\frac{1}{\lambda_{i, t}^C}) v_{i, t} \to \infty$.

\begin{lemma}\label{lemma:mul-conv}
Under Assumption \ref{ass:highest-value-assumption}, for any $i \in [N], t \in[T]$, with probability at least $1 - \exp\big(-\frac{\Delta^2 t}{32 B^2 N^2}\big)$, $h'(\lambda_{i, t+1}^C) \le h'(\lambda_{i, 1})-\frac{\Delta t}{2N \sqrt T}$. 
% $\lambda_{i, t+1}^C \le \lambda_{i, 1} \exp\big(-\frac{\Delta t}{2N \sqrt T}\big)$.
\end{lemma}

\begin{proof}
% \gyr{Copy proof of Lemma 4.1. The only difference is $G(\lambda_t)\ge \Delta /n$. Can put in appendix.}
Let $y_{i,t}:= h'(\lambda_{i,t}^C)$. According to Algorithm \ref{alg:general-mirror-descent} and Equation (\ref{eq:md-y-update}),
\begin{equation}
y_{i,t+1}=y_{i,t}-\alpha g_{i,t}
\end{equation}
% According to Algorithm \ref{alg:feng-algorithm}, $\lambda_{i, t+1}^C$ is equal to 
% \[
    % \lambda_{i, t+1}^C = \lambda_{i, 1} \exp\Big( - \alpha \sum_{t'=1}^t g_{i, t'} \Big)
% \]
where each $g_{i, t}$ is an unbiased estimator of bidder $i$'s expected utility $G_i^C(\lambda_{i, t}^C) = \E[ g_{i, t}]$ at round $t$. 
If bidder $i$ is not the highest-value bidder in round $t$, then $g_{i, t} = 0$, meaning that the bidder's multiplier $\lambda_{i, t}^C$ is not updated in that round. 
According to Observation \ref{observation:utility-1/N} and Lemma \ref{lem:utility-at-least-Delta}, bidder $i$'s expected utility is at least 
\[
    \E[ g_{i, t'}] = G_i^C(\lambda_{i, t'}^C) = \tfrac{1}{N} G_{(N)}(\lambda_{i, t'}) \ge \tfrac{\Delta}{N}. 
\]
Because values and outside bids are bounded by $B$, the realized utility $g_{i, t'} \in [-B, B]$. So, by Azuma's inequality:
\begin{align*}
\Pr\left[\sum \nolimits_{t'=1}^t g_{i, t'}<t \cdot \tfrac{\Delta}{N} - \eps \right]
% \le \Pr\left[\sum_{t'=1}^t g_{t'}<\sum_{t'=1}^t \E[g_{t'}]- \epsilon\right]
\le \exp\Big(-\tfrac{\eps^2}{8B^2t}\Big). 
\end{align*}
Let $\eps = \frac{t\Delta}{2N}$. We have with probability at least $1 - \exp( - \frac{t \Delta^2}{32 B^2 N^2})$,
\begin{align*}
    y_{i, t+1} \le y_{i, 1} -\alpha \left( t\cdot \tfrac{\Delta}{N}-\eps\right)  = y_{i, 1}  -\tfrac{t\Delta}{2 N \sqrt T}.
\end{align*}
where we plugged in $\alpha =1 / \sqrt{T}$. 
\end{proof}

\paragraph{\bf Coordinated total value is larger than independent total value.}
Using the fact that coordinated multipliers converge to $0$, we show that the total value obtained by coordinated auto-bidders is larger than that under independent bidding in the limit $T\to\infty$.

Let $V_{(N)}(\lambda)$ be the expected value obtained by a single bidder whose value is the highest value $v_{(N)}$ among $N$ samples from $F$ and who uses multiplier $\lambda > 0$ to bid against outside bid $d^O \sim D$: 
% \[
%     V_{(N)}(\lambda) = \E_{v_{(N)} \sim F_{(N)}, d^O\sim D}\Big[ v_{(N)} \cdot \1\big[ (1+\tfrac{1}{\lambda}) v_{(N)} > d^O \big] \Big]. 
% \]
\begin{equation}
V_{(N)}(\lambda) ~ := ~ \mathbb{E}_{v_{(N)} \sim F_{(N)}, d^O\sim D}\Bigl[v_{(N)}\1\left[ (1+ \tfrac{1}{\lambda} )v_{(N)} > d^O\right]\Bigr]. \label{eq:Vdef}
\end{equation}
Because $V_{(N)}(\lambda)$ is monotone in $\lambda$ and bounded, the limit
\[
    V_{(N)}(0^+) ~ := ~ \lim_{\lambda \to 0^+} V_{(N)}(\lambda)
\]
exists. Note that the expected total value of the $N$ coordinated bidders with multipliers $\lambda_1^C, \ldots, \lambda_N^C$ is
\begin{align*}
    V^C(\lambda_1^C, \ldots, \lambda_N^C) ~ = ~ \E\Big[ v_{i^*} \1\big[ (1+ \tfrac{1}{\lambda_{i^*}}) v_{i^*} > d^O \big]\Big] ~ = ~ \E\big[ V_{(N)}(\lambda_{i^*}^C )\big]
\end{align*}
where $i^* = \argmax_{i\in[N]} v_i$. This means that the expected average total value of the coordinated bidders during $T$ rounds is
\begin{align}
    \E \Big[ \frac{1}{T} \sum_{t=1}^T V_t^{C,A} \Big] ~ = ~ \E \Big[ \frac{1}{T} \sum_{t=1}^T V_{(N)}(\lambda_{i^*_t, t}^C) \Big]  
\end{align}
where $i^*_t\in[N]$ is the highest-value bidder in round $t$.

\begin{lemma}
\label{lemma:lim_ge_ind}
For any coordination mechanism $\mathcal G$ and any round $t$,
\[
\E\!\left[V_t^{\mathcal G}\right] ~ \le ~ V_{(N)}(0^+) ~ = ~ \E_{v_{(N)}\sim F_{(N)}}\big[v_{(N)}\big].
\]
Consequently, for every $T\ge 1$,
\[
\E\Big[\frac{1}{T}\sum_{t=1}^T V_t^{\mathcal G}\Big] ~ \le ~ V_{(N)}(0^+).
\]
\end{lemma}
\begin{proof}
Fix a round $t$ and write $v_{(N),t}\coloneqq \max_{i\in[N]} v_{i,t}$.
In a single-item auction, the coalition's realized value is either $0$ (if the outside bid wins) or equals the realized value of the winning coalition bidder, which is always at most $v_{(N),t}$. Hence
$V_t^{\mathcal G} \le v_{(N),t}$ almost surely, and taking expectations gives
$\E[V_t^{\mathcal G}] \le \E[v_{(N),t}]$.
Finally, since bids and $d^O$ are bounded, $(1+\tfrac{1}{\lambda})v_{(N)} \to +\infty$ as $\lambda\downarrow 0$, so the single bidder wins with probability $1$ in the limit and thus
$V_{(N)}(0^+) = \E[v_{(N)}]$.
\end{proof}

% \begin{lemma}
% \label{lemma:lim_ge_ind}
% $V_{(N)}(0^+) \ge V^I(\lambda_1, \ldots, \lambda_N)$, for any $\lambda_1, \ldots, \lambda_N > 0$.
% \end{lemma}
% \begin{proof}
% Pick $\lambda_\eps < \min\{\lambda_1, \dots, \lambda_N\}$. 
% If some bidder $i \in [N]$ wins the item under independent bidding, then their bid $b_i^I = (1+\frac{1}{\lambda_i}) v_i$ is larger than the outside bid $d^O$. 
% In that case, suppose the highest-value bidder uses multiplier $\lambda_\eps$, then their bid is $(1+\frac{1}{\lambda_\eps}) v_{(N)} \ge (1+\frac{1}{\lambda_i}) v_i = b_i^I > d^O$, so the highest-value bidder also wins in the single-bidder scenario. Note that the value obtained by the highest-value bidder is at least $v_i$, so the expected value $V_{(N)}(\lambda_\eps) \ge V^I(\lambda_1, \ldots, \lambda_N)$.  Letting $\lambda_\eps \to 0$ proves the lemma. 
% \end{proof}

% \gyr{$V_{(N)}(0^+)=\E_{v_{(N)} \sim F_{(N)}}[v_{(N)}]$, so here we can replace $V^I(\lambda_1, \ldots, \lambda_N)$ with any coordination algorithm, and indeed prove that our coordination algorithm is asymptotically optimal in value.}

\begin{proof}[\bf Proof of Theorem \ref{thm:value}]
By Lemma~\ref{lemma:mul-conv}, with constant
$c \coloneqq \frac{\Delta^2}{32B^2N^2}$, for any fixed $i$ and $t\ge 1$, the coordinated multiplier $\lambda_{i, t+1}^C$ satisfies 
\begin{equation}\label{eq:pt}
\Pr\!\left[h'(\lambda_{i,t+1}^C) \le h'(\lambda_{i,1})-\frac{t\Delta}{2N\sqrt{T}}\right]\;\ge\;1-e^{-ct}.
\end{equation}
Fix $\varepsilon>0$. Since $V_{(N)}(0^+)$ exists, there is $\delta=\delta(\varepsilon)>0$ such that
\begin{equation}\label{eq:cont}
\big| V_{(N)}(0^+)-V_{(N)}(\lambda) \big| <\varepsilon \quad \text{for all } \lambda \in (0, \delta].
\end{equation}
Since $h'$ is strictly increasing and $\lim_{\lambda\downarrow 0}h'(\lambda)=-\infty$, we can define $b(\delta):=(h')^{-1}(\delta)$.

Define
\begin{equation*}\label{eq:mT}
m_T \;\coloneqq\; \left\lceil \frac{2N\sqrt{T}}{\Delta}(h'(\lambda_{i,1})-b(\delta))\right\rceil ~ = ~ O\Big(\sqrt{T} (-b(\delta))\Big) ,
\end{equation*}
so for all $t\ge m_T$,
\[
h'(\lambda_{i,1})-\frac{t\Delta}{2N\sqrt{T}}\le
h'(\lambda_{i,1})-\frac{m_T\Delta}{2N\sqrt{T}}\le b(\delta) .
\]
Let $E_{i,t}$ be the event inside the probability in (\ref{eq:pt}). A union bound over $i=1,\dots,N$ and $t\ge m_T$ gives
\begin{equation}\label{eq:hp}
\Pr\!\Big[\bigcap_{i=1}^N\bigcap_{t=m_T}^{T} E_{i,t}\Big]
% \;\ge\; 
% 1 - \sum_{i=1}^N\sum_{t=m_T}^{\infty}e^{-ct}
% \;=\; 1 - N\,\frac{e^{-cm_T}}{1-e^{-c}}
\; \ge \; 1 - N T e^{-cm_T}
\;\eqqcolon\; 1-p_T .
\end{equation}
On this event, for all $t\ge m_T$ we have $h'(\lambda_{i_t^*,t}^C)\le b(\delta)$, which is equivalent to $\lambda^C_{i^*_t,t}\le\delta$, hence by Eq.~(\ref{eq:cont}),
\[
\big| V_{(N)}(0^+)-V_{(N)}(\lambda^C_{i^*_t,t}) \big| <\varepsilon .
\]
For the first $m_T$ rounds we have $| V_{(N)}(0^+)-V_{(N)}(\cdot) |\le B$ due to bounded value. Therefore, conditioning on the event in (\ref{eq:hp}),
\begin{equation*}\label{eq:hp-gap}
\Big| V_{(N)}(0^+) - \frac{1}{T}\sum_{t=1}^T V_{(N)}(\lambda^C_{i^*_t,t}) \Big|
\;\le\; \frac{m_T}{T}\,B + \frac{T-m_T}{T}\,\varepsilon. 
\end{equation*}
% Taking expectations and using that the per-round gap is $\le B$ on the complement of the event in (\ref{eq:hp}),
% \begin{equation}\label{eq:exp-gap}
% \E\Big[\Big| V_{(N)}(0^+) - \frac{1}{T} \sum_{t=1}^T V_{(N)}(\lambda^C_{i^*_t,t}) \Big| \Big]
% \;\le\; \frac{m_T}{T}\,B + \frac{T-m_T}{T}\,\varepsilon' + B\,p_T,
% \end{equation}
% where % $p_T = N\,\dfrac{e^{-cm_T}}{1-e^{-c}}$.
% $p_T = N T e^{-cm_T} = NT e^{-c O(\sqrt{T})}$. 
% 
% Substituting \eqref{eq:mT} yields
% \[
% c\,m_T \;\ge\; \frac{\Delta^2}{32B^2N^2}\cdot \frac{2N\sqrt{T}}{\Delta}\,\log\!\Big(\frac{\lambda_{i,1}}{\delta}\Big)
% \;=\; \frac{\Delta}{16B^2N}\,\sqrt{T}\,\log\!\Big(\frac{\lambda_{i,1}}{\delta}\Big),
% \]
% so there exist constants $C_1,C_2>0$ (depending only on $B,N,\Delta,\lambda_{i,1},\delta$) with
% \begin{equation}\label{eq:tail}
% p_T \;\le\; C_1 e^{-C_2\sqrt{T}}.
% \end{equation}
% Since $m_T=O(\sqrt{T})$, the right-hand side of \eqref{eq:exp-gap} tends to $\varepsilon'$ as $T\to\infty$, and letting $\varepsilon'\downarrow 0$ gives
% \[
% \frac{1}{T}\,\E\Big[\sum_{t=1}^T V_{(N)}(\lambda^C_{i^*_t,t})\Big] \;\longrightarrow\; V_{(N)}(0^+).
% \]
% Combining Lemma~\ref{lemma:lim_ge_ind} with the equation above yields Theorem~\ref{thm:value}. 
% 
Taking expectations and using that the per-round gap is $\le B$ on the complement of the event in (\ref{eq:hp}), we have
\begin{align*}
& \E\bigg[\Big| V_{(N)}(0^+) - \frac{1}{T} \sum_{t=1}^T V_{(N)}(\lambda^C_{i^*_t,t}) \Big| \bigg]\\
& \;\le\; \frac{m_T}{T}\,B + \frac{T-m_T}{T}\,\varepsilon + B\,p_T \\
& \; \le\; \frac{O(\sqrt{T} (-b(\delta)))}{T}\,B + \varepsilon + B NT e^{-c O(\sqrt{T} (-b(\delta)))} ~ \le ~ 3\eps
\end{align*}
for sufficiently large $T$.  So we obtain
\[
% \lim_{T\to\infty} \E\Big[ \frac{1}{T}\sum_{t=1}^T V^C_t \Big] \, = \, 
\lim_{T\to\infty} \E\Big[ \frac{1}{T}\sum_{t=1}^T V_{(N)}(\lambda^C_{i^*_t,t})\Big] \, =\, V_{(N)}(0^+).
\]
Combining with Lemma~\ref{lemma:lim_ge_ind}, we conclude that
\begin{align*}
    V_{(N)}(0^+) \; \ge \; \E\Big[ \frac{1}{T} \sum_{t=1}^T V_t^{\mathcal G}\Big]  , 
\end{align*}
which proves Theorem~\ref{thm:value}. 
\end{proof}

\section{Omitted Proofs in Section \ref{sec:asymmetric}}

\subsection{Proof of Theorem~\ref{thm:asym-violation}}
\label{app:thm-asym-vio}
The argument parallels the proof of Theorem~\ref{thm:utility/RoS-constraint}.

Let $i^*_t = \arg\max_{i \in [N]} v_{i,t}$ denote the index of the highest-value bidder in round~$t$. 
Conditioned on the history $H_{t-1}$, we analyze the expected difference $\Delta_{i,t} = u_{i,t}^{C,A} -u_{i,t}^{\truth}$ for bidder~$i$:
\begin{align*}
& \E_{\bm v_t \sim F,\, d_t^O \sim D}\big[\Delta_{i,t} \mid H_{t-1}\big]\\
&= \E_{\bm v_t,\, d_t^O}\big[\1[i^*_t = i] \cdot \Delta_{i,t} \mid H_{t-1}\big] \\
&= \E\Big[\1[i^*_t = i] \cdot \1[v_{i,t} \ge d_t^O] \cdot (d_t^{\truth} - d_t^O)\Big] + \E\Big[\1[i^*_t = i] \cdot \1[v_{i,t} \le d_t^O \le b_{i,t}^{C,A}] \cdot (v_{i,t} - d_t^O)\Big]\\
&= \E\Big[\1[i^*_t = i] \cdot (d_t^{\truth} - d_t^O)\Big] + \E\Big[\1[i^*_t = i] \cdot \1[v_{i,t} \le d_t^O \le b_{i,t}^{C,A}] \cdot (v_{i,t} - d_t^O)\Big]\\
&\ge \E\Big[\1[i^*_t = i] \, \big(\max_{j \ne i} v_{j,t} - d_t^O \big)_+\Big] + \E\Big[\1[i^*_t = i] \cdot \1[v_{i,t} \le d_t^O]\cdot (v_{i,t} - d_t^O)\Big].
\end{align*}

Summing over all bidders $i \in [N]$, the indicator $\1[i^*_t = i]$ ensures that only the highest-value bidder contributes in each round. 
Hence,
\begin{align*}
\E\Big[ \sum_{i=1}^N \Delta_{i,t} \; \big| \;  H_{t-1} \Big]
&\ge \E\big[(v_{(N-1)} - d_t^O)_+ \big] - \E\big[(d_t^O - v_{(N)})_+\big] \\
&\ge \Delta,
\end{align*}
where the second inequality follows from Assumption~\ref{ass:highest-value-assumption}.
% Finally, applying Azuma’s inequality (in a similar way as in the proof of Theorem \ref{thm:utility/RoS-constraint}) completes the proof.
% \tao{Is $\sum_{i=1}^N \Delta_{i, t}$ bounded by $[-B, B]$?}
Summing over $t \in [T]$ gives $\E[\sum_{i=1}^N U_{i, T}^{C, A} - \sum_{i=1}^N U_{i, T}^{I, A}] \ge  \Delta T$. 

\subsection{Proof of Theorem~\ref{thm:asym-value}}
\label{app:thm-asym-value}

The argument parallels the proof of Theorem~\ref{thm:value}. 

Let $G_i^C(\lambda_i)$ denote bidder $i$’s expected utility when $v_i\sim F_i$ and bidder $i$ uses multiplier $\lambda_i>0$ in the coordinated scenario:
\begin{align*}
    G_i^C(\lambda_i)
    \;=\;
    \E_{v_1\sim F_1,\dots,v_N\sim F_N,\; d^O\sim D}
    \big[(v_i-d^O)\,x_i(\bm v,\lambda_i)\big],
\end{align*}
where $x_i(\bm v,\lambda_i)=1$ if $v_i=\max_{j\in[N]}\{v_j\}$ and $(1+\tfrac{1}{\lambda_i})v_i>d^O$, and $0$ otherwise.

By reusing the proof of Lemma~\ref{lemma:mono} with $D$ replaced by the distribution of $\max\{d^O,\max_{j\neq i} v_j\}$, we claim that $G_i^C(\lambda_i)$ is \emph{increasing} in $\lambda_i$. Hence,
\begin{align*}
    G_i^C(\lambda_i) ~ \ge ~
    \lim_{\mu\to 0^+} G_i^C(\mu) ~ = ~ E\Big[(v_i-d^O)\,\1\big\{v_i=\max_{j\in[N]}\{v_j\}\big\}\Big] ~ =: ~ \Delta_i,
\end{align*}
where the limit exists because $G_i^C(\cdot)$ is bounded and monotone.

\medskip
We now present the parallel version of Lemma~\ref{lemma:mul-conv}.

\begin{lemma}\label{lemma:asym-mul-conv}
Under Assumption~\ref{ass:non-iid-value}, for any $i\in[N]$ and $t\ge 1$, with probability at least $1-\exp\!\big(-\tfrac{\Delta_i^2 t}{32B^2}\big)$,
% \[
% \lambda_{i,t+1}^C \;\le\; \lambda_{i,1}\,\exp\!\Big(-\frac{\Delta_i\, t}{2\sqrt{T}}\Big).
% \]
\[ h'(\lambda_{i, t+1}^C) \le h'(\lambda_{i, 1})-\frac{\Delta_i t}{2 \sqrt T}. \]
\end{lemma}

\begin{proof}
The expected utility $\E[g_{i,t'}]$ is lower bounded by
\[
\E[g_{i,t'}] \;=\; G_i^C(\lambda_{i,t'}) \;\ge\; \Delta_i.
\]
Applying Azuma’s inequality (exactly as in the proof of Lemma~\ref{lemma:mul-conv}) completes the argument.
\end{proof}

Write $v_{(N)}\coloneqq \max_{j\in[N]} v_j$ with $v_j\sim F_j$ mutually independent. 
Its CDF is $F^{\max}(x)=\Pr[v_{(N)}\le x]=\prod_{j=1}^N F_j(x)$. 
For $\lambda>0$, define
\begin{align*}
V_{(N)}(\lambda) & ~ \coloneqq ~
\E_{\,v_{(N)}\sim F^{\max},\, d^O\sim D}\Big[\,v_{(N)}\cdot \1\big\{(1+\tfrac{1}{\lambda})v_{(N)} > d^O\big\}\Big],\\
V_{(N)}(0^+) & ~ \coloneqq ~ \lim_{\lambda\to 0^+} V_{(N)}(\lambda) ~ = ~ \E_{\,v_{(N)}\sim F^{\max}} \big[\,v_{(N)} \big].
\end{align*}
% Since the indicator increases pointwise to $\1\{v_{(N)}>d^O\}$ as $\lambda \to 0^+$ and the integrand is bounded by $B$, dominated (or monotone) convergence implies $V_{(N)}(\lambda)\to V_{(N)}(0^+)$ as $\lambda \to 0^+$.
% ; thus the continuity condition used in \eqref{eq:cont} holds with this definition. 
The limit $V_{(N)}(0^+) = \lim_{\lambda\to 0^+} V_{(N)}(\lambda)$ exists because $V_{(N)}(\lambda)$ is monotone and bounded. 
In the i.i.d.\ special case $F_1=\cdots=F_N=F$, $F^{\max}$ coincides with the order-statistic distribution $F_{(N)}$, recovering the original definition.

In the non-i.i.d.~case, Lemma \ref{lemma:lim_ge_ind} still holds: 
\begin{lemma}[copy of Lemma \ref{lemma:lim_ge_ind}]
\label{lemma:lim_ge_ind-non-iid}
% $V_{(N)}(0^+) \ge V^I(\lambda_1, \ldots, \lambda_N)$.
For any coordination mechanism $\mathcal G$ and any round $t$,
\[
\E\!\left[V_t^{\mathcal G}\right] ~ \le ~ V_{(N)}(0^+) ~ = ~ \E_{v_{(N)}\sim F^{\max}}\big[v_{(N)}\big].
\]
Consequently, for every $T\ge 1$,
\[
\E\Big[\frac{1}{T}\sum_{t=1}^T V_t^{\mathcal G}\Big] ~ \le ~ V_{(N)}(0^+).
\]
\end{lemma}

% \begin{proof}[\bf Proof of Lemma~\ref{lemma:lim_ge_ind} under non-i.i.d.\ case]
% Let $\lambda_\varepsilon < \min\{\lambda_1,\dots,\lambda_N\}$. 
% If bidder $i$ wins under independent bidding, then $(1+\tfrac{1}{\lambda_i})v_i > d^O$. 
% Since $v_{(N)}\ge v_i$, we have
% \[
% (1+\tfrac{1}{\lambda_\varepsilon})\,v_{(N)}
% \;\ge\;
% (1+\tfrac{1}{\lambda_i})\,v_i
% \;>\; d^O,
% \]
% so the single coordinated bidder with multiplier $\lambda_\varepsilon$ (who bids against $d^O$ using $v_{(N)}$) also wins and realizes value at least $v_i$. 
% Taking expectations over $(v_1,\dots,v_N,d^O)$ yields
% $V_{(N)}(\lambda_\varepsilon)\ge V^I(\lambda_1,\dots,\lambda_N)$.
% Letting $\lambda_\varepsilon\downarrow 0$ completes the proof.
% \end{proof}

\begin{proof}[\bf Proof of Theorem~\ref{thm:asym-value}]
Fix $\varepsilon>0$ and choose $\delta=\delta(\varepsilon)>0$ as in Equation (\ref{eq:cont}). 
Let
\[
\Delta_{\min}\coloneqq \min_{i\in[N]}\Delta_i > 0 \qquad 
c_{\min}\coloneqq \frac{\Delta_{\min}^2}{32B^2} > 0. 
\]
Define $b(\delta):=(h')^{-1}(\delta)$ and
\[
m_T \;\coloneqq\; \Big\lceil \frac{2\sqrt{T}}{\Delta_{\min}}(h'(\lambda_{i,1})-b(\delta))\Big\rceil .
\]
By Lemma~\ref{lemma:asym-mul-conv}, for each $i$ and every $t\ge m_T$,
\(
\lambda_{i,t+1}^C \le \delta
\)
holds with probability at least $1-\exp(-c_{\min} t)$.
A union bound over $i=1,\dots,N$ and $t=m_T,\dots,T$ yields
\[
\Pr\!\Big[\bigcap_{i=1}^N\bigcap_{t=m_T}^{T}\{\lambda_{i,t+1}^C\le \delta\}\Big]
\;\ge\; 1 - N T\, e^{-c_{\min} m_T}
\;\eqqcolon\; 1-p_T .
\]
On this event, for all $t\ge m_T$ we have $\lambda^C_{i^*_t,t}\le\delta$, hence by (\ref{eq:cont}),
\[
\big| V_{(N)}(0^+)-V_{(N)}(\lambda^C_{i^*_t,t}) \big| < \varepsilon,
\]
where $i^*_t=\arg\max_{i\in[N]} v_{i,t}$. For each of the first $m_T$ rounds, the absolute gap is $\le B$. Therefore,
\[
\Big| V_{(N)}(0^+) - \frac{1}{T}\sum_{t=1}^T V_{(N)}(\lambda^C_{i^*_t,t}) \Big|
\;\le\; \frac{m_T}{T}\,B + \frac{T-m_T}{T}\,\varepsilon
\]
on the above high-probability event; on its complement the gap is $\le B$. 
Taking expectations and using $p_T \le N T e^{-c_{\min} m_T}$ gives
\[
\E\Big[\Big| V_{(N)}(0^+) - \frac{1}{T}\sum_{t=1}^T V_{(N)}(\lambda^C_{i^*_t,t}) \Big|\Big]
\;\le\; \frac{m_T}{T}B + \varepsilon + B\,p_T \;\xrightarrow[T\to\infty]{}\; 0.
\]
Hence
\[
\lim_{T\to\infty}\E\Big[\frac{1}{T}\sum_{t=1}^T V_{(N)}(\lambda^C_{i^*_t,t})\Big]
~ = ~ V_{(N)}(0^+).
\]
Finally, combining with Lemma~\ref{lemma:lim_ge_ind-non-iid} yields the desired inequality against the independent-bidding total value, completing the proof.
\end{proof}

\section{Additional Experiment Results}
\label{app:fig}
This appendix presents additional figures in the experiments: Figure \ref{fig:iid} for the symmetric setting on synthetic data; Figure \ref{fig:non-iid} for the asymmetric setting on synthetic data; Figure \ref{fig:REAL-MERGE-K5} for the real-world dataset. 

\begin{figure}[H]
  \centering
  \begin{subfigure}[b]{0.48\textwidth}
    \centering
    \includegraphics[width=\textwidth]{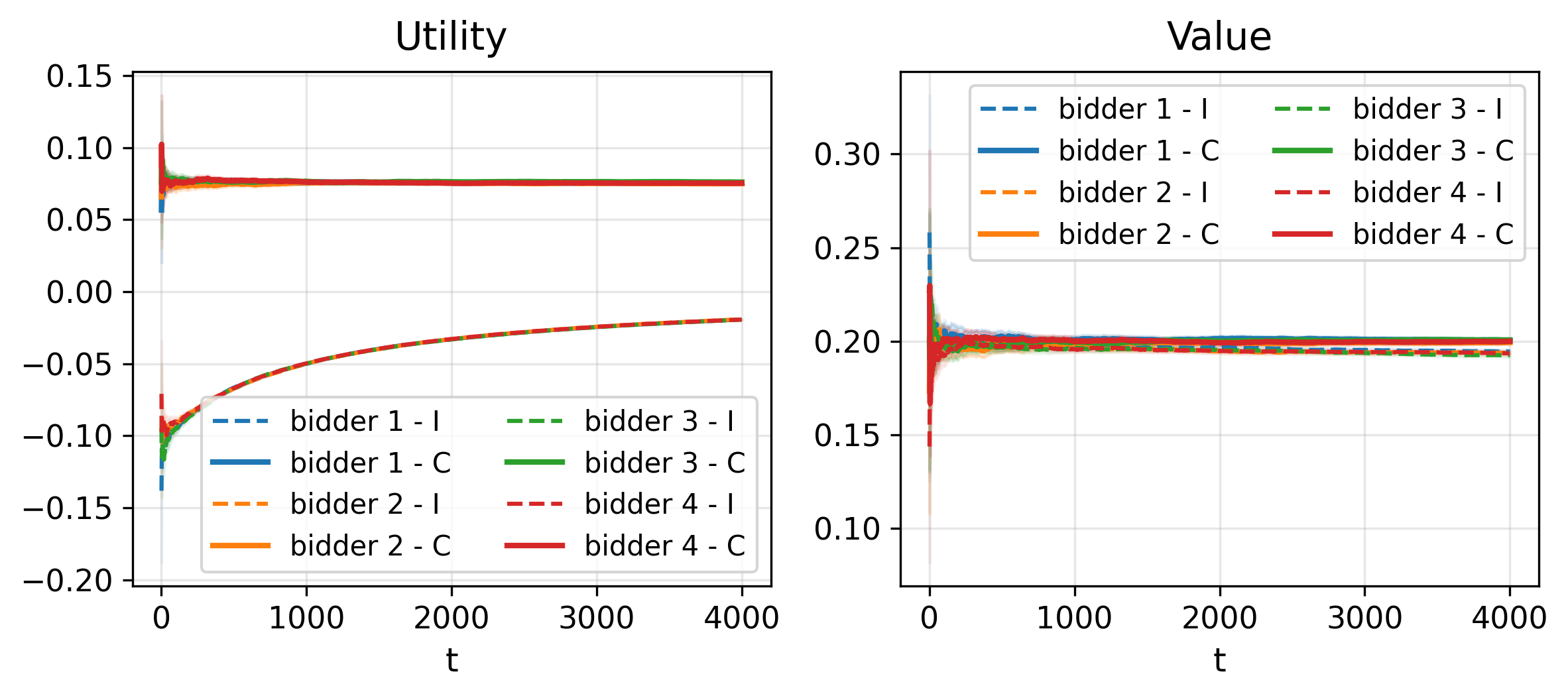}
    \caption{$N=4, F=U[0,1], D=U[0,1]$}
    \label{fig:I2}
  \end{subfigure}
  \hfill
  \begin{subfigure}[b]{0.48\textwidth}
    \centering
    \includegraphics[width=\textwidth]{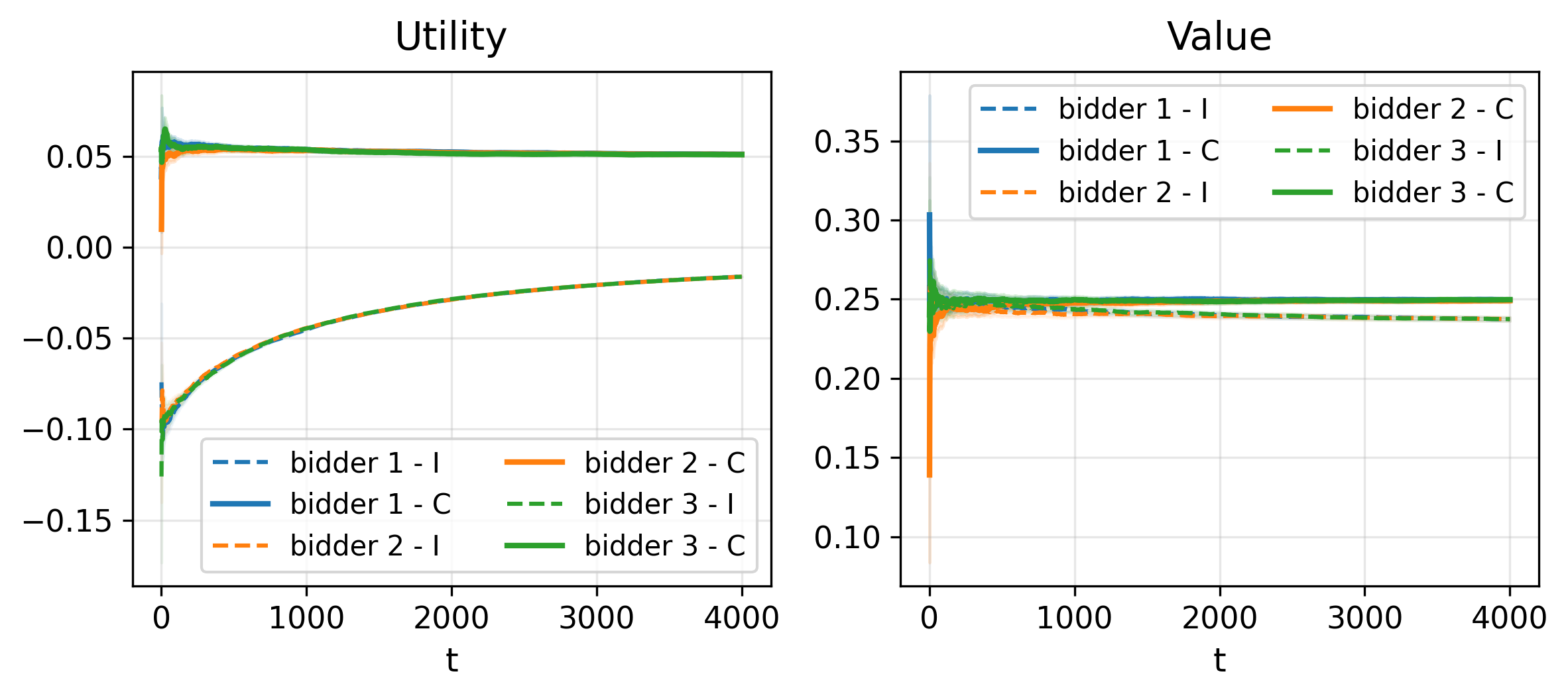}
    \caption{$N=3, F=U[0,1], D=\mathrm{Beta}(3,2)$}
    \label{fig:I3}
  \end{subfigure}
  \caption{Experiments under i.i.d auto-bidders.}
  \label{fig:iid}
\end{figure}

\begin{figure}[H]
  \centering
  \begin{subfigure}[b]{0.48\textwidth}
    \centering
    \includegraphics[width=\textwidth]{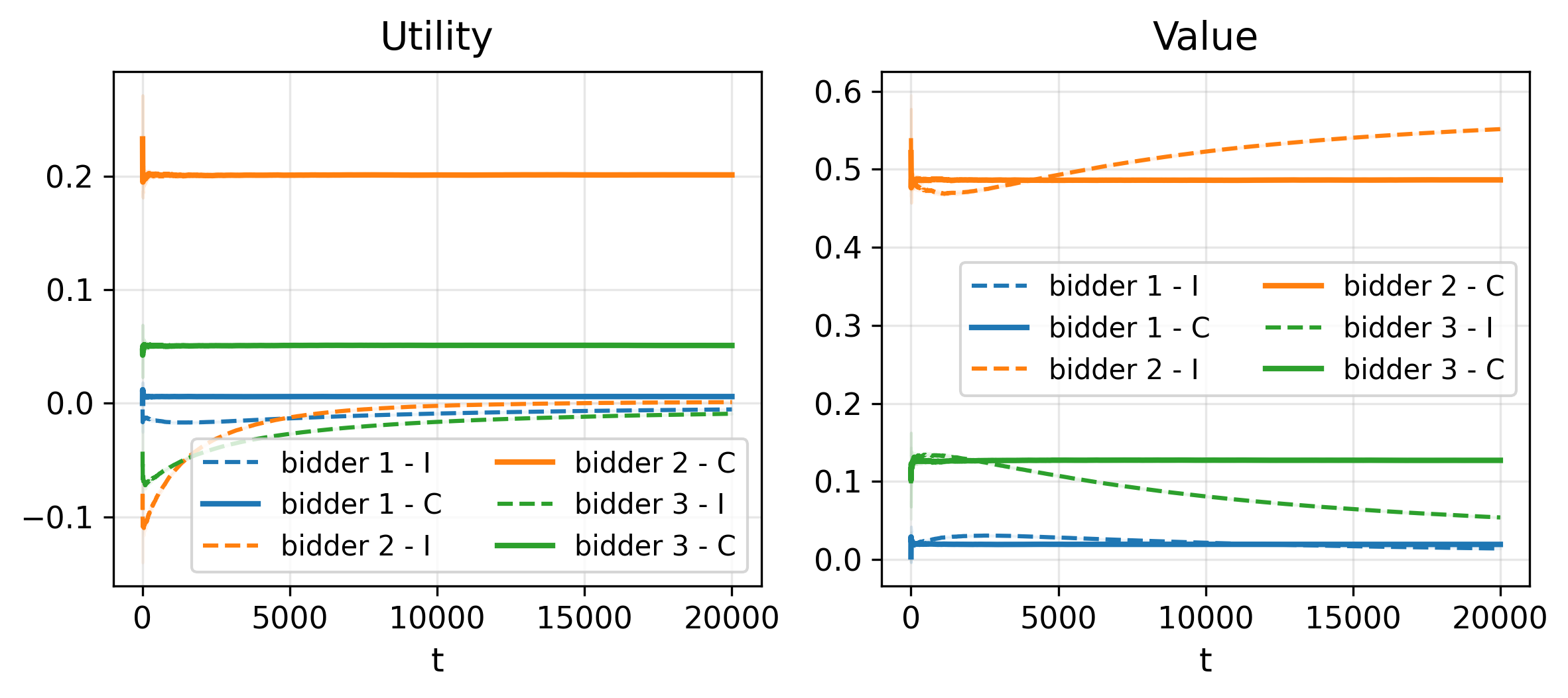}
    \caption{$N=3,F_1=\mathrm{Beta}(2, 6),F_2=\mathcal{N}_{\mathrm{cap}=[0,1]}(0.6,0.15),F_3=\mathcal{N}_{\mathrm{cap}=[0,1]}(0.4,0.2),D=\mathrm{Beta}(3,5)$}
    \label{fig:NI2}
  \end{subfigure}
  \hfill
  \begin{subfigure}[b]{0.48\textwidth}
    \centering
    \includegraphics[width=\textwidth]{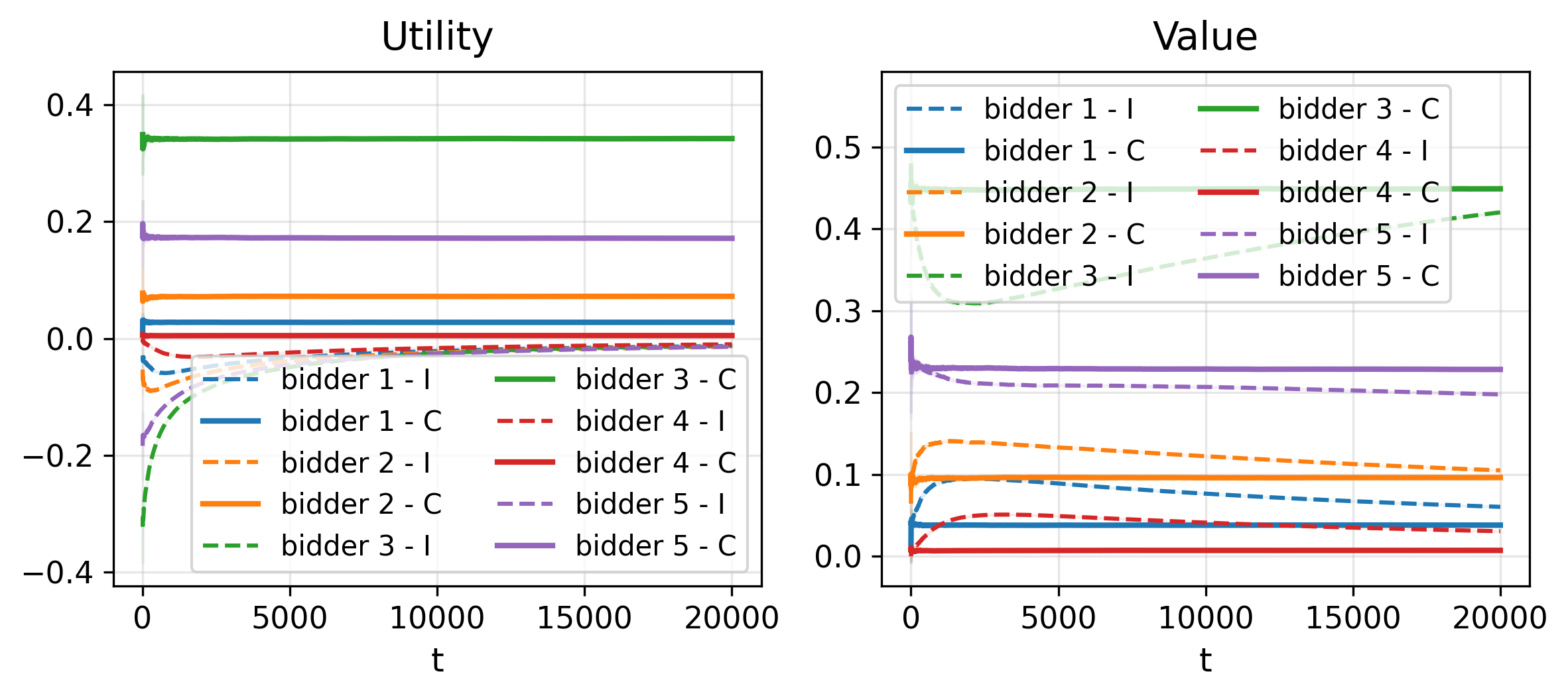}
    \caption{$N=5, F_1=U[0.2,0.8],F_2=\mathrm{Beta}(4,3), F_3=\mathrm{Beta}(6,2), F_4=\mathcal{N}_{\mathrm{cap}=[0,1]}(0.5,0.1),F_5=\mathcal{N}_{\mathrm{cap}=[0,1]}(0.7,0.12), D=\mathrm{Beta}(2,8)$}
    \label{fig:NI3}
  \end{subfigure}
  \caption{Experiments under non-i.i.d auto-bidders.}
  \label{fig:non-iid}
\end{figure}

\begin{figure}[H]
  \centering
  \includegraphics[width=0.48\textwidth]{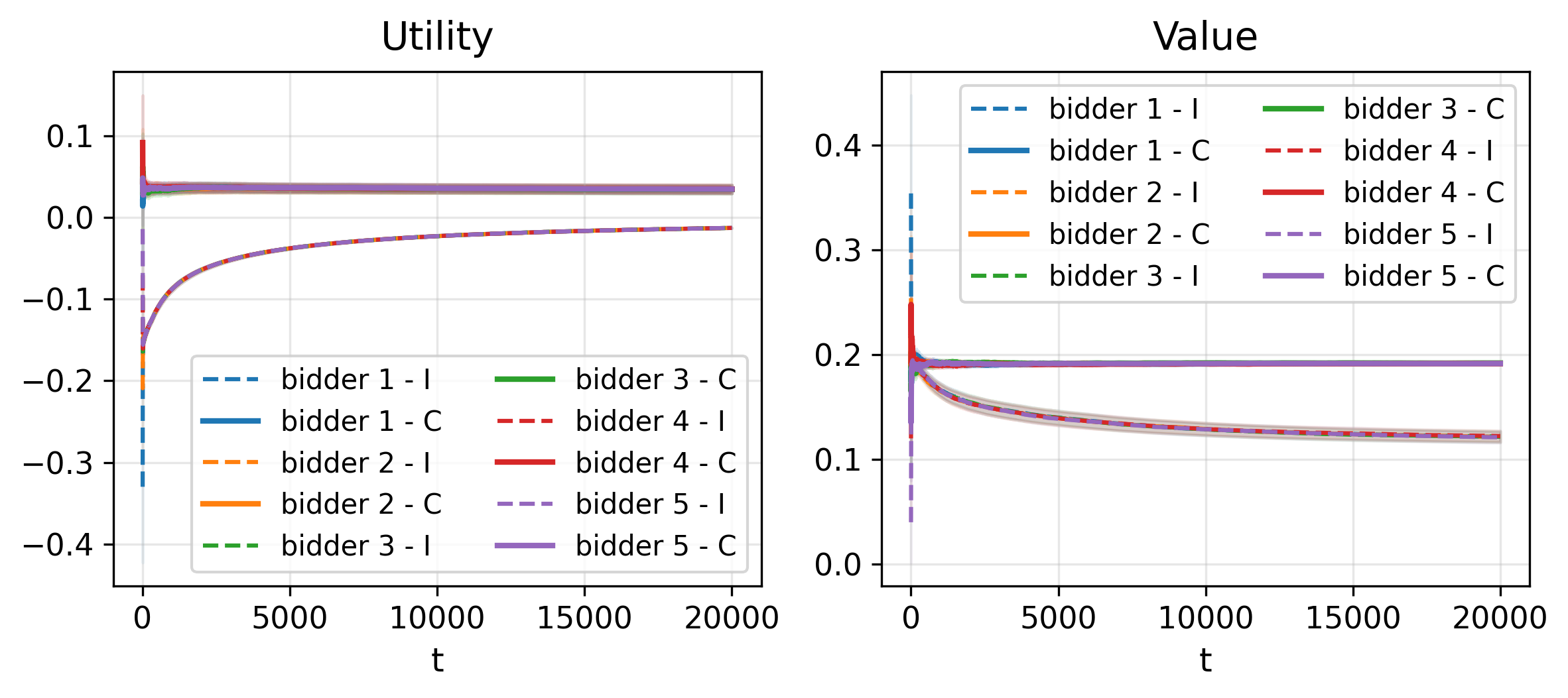}
  \caption{Experiments in real-world datasets. ($N=5$).}
  \label{fig:REAL-MERGE-K5}
\end{figure}

\end{document}